\documentclass[12pt,draftclsnofoot,journal,letter,onecolumn]{IEEEtran}

\usepackage{amsmath,amssymb}
\usepackage{cuted}
\usepackage{bbm}
\newcommand{\zbm}{\textbf{0}}
\newcommand\norm[1]{\left\lVert#1\right\rVert}

\newcommand{\RNum}[1]{\uppercase\expandafter{\romannumeral #1\relax}}
\usepackage{graphicx}
\usepackage{epsfig}
\usepackage{latexsym}
\usepackage{amsfonts}
\usepackage{here}
\usepackage{rawfonts}
\usepackage[latin1]{inputenc}
\usepackage[T1]{fontenc}
\usepackage{calc}
\usepackage{url}
\usepackage{enumerate}
\usepackage{color}
\usepackage{upref}
\usepackage{epic,eepic}
\usepackage{times}
\usepackage{dsfont}
\usepackage{comment}
\usepackage{cite}

\newtheorem{theorem}{{\bf Theorem}}

\newcommand{\qed}{\nobreak \ifvmode \relax \else
  \ifdim\lastskip<1.5em \hskip-\lastskip
  \hskip1.5em plus0em minus0.5em \fi \nobreak
  \vrule height0.75em width0.5em depth0.25em\fi}

\newcounter{step}
\newlength{\totlinewidth}
  {\end{list}%
  \rule{\linewidth}{1pt}}
\newcounter{substep}

  {\end{list}}

\newlength{\aligntop}
\newlength{\alignbot}
 

\IEEEoverridecommandlockouts

\begin{document}

\title{Asynchronous Downlink Massive MIMO Networks: A Stochastic Geometry Approach}

\author{Elahe Sadeghabadi, 
Seyed Mohammad Azimi-Abarghouyi, 
Behrooz Makki, 
Masoumeh Nasiri-Kenari,~\IEEEmembership{Senior Member, IEEE}
\thanks{E. Sadeghabadi, S.M. Azimi-Abarghouyi, and M. Nasiri-Kenari
are with Electrical Engineering Department, Sharif University of Technology, Tehran, Iran. Emails: \textit{sadeghabadi108@yahoo.com}, \textit{azimi$\_$sm@ee.sharif.edu}, and \textit{mnasiri@sharif.edu}. 
B. Makki is with Department of Electrical Engineering, Chalmers University of Technology, Gothenburg, Sweden, Email: \textit{behrooz.makki@chalmers.se}
}
}

\maketitle
\begin{abstract}
Massive multiple-input multiple-output (MIMO) is recognized as a promising technology for the next generation of wireless networks because of its potential to increase the spectral efficiency. In initial studies of massive MIMO, the system has been considered to be perfectly synchronized throughout the entire cells. However, perfect synchronization may be hard to attain in practice. Therefore, we study a massive MIMO system whose cells are not synchronous to each other, while transmissions in a cell are still synchronous. We analyze an asynchronous downlink massive MIMO system in terms of the coverage probability and the ergodic rate by means of the stochastic geometry tool. For comparison, we also obtain the results for the synchronous systems. In addition, we investigate the effect of the uplink power control and the number of pilot symbols on the downlink ergodic rate, and we observe that there is an optimal value for the number of pilot symbols maximizing the downlink ergodic rate of a cell. Our results also indicate that, compared to the cases with synchronous transmission, the downlink ergodic rate is more sensitive to the uplink power control in the asynchronous mode.
\end{abstract}
%\begin{keywords}
%Massive MIMO, downlink, synchronous and asynchronous systems, stochastic geometry, coverage probability.
%\end{keywords}
\IEEEpeerreviewmaketitle
\section{Introduction}
The need for a higher data rate is getting a vital factor in the next generation of wireless networks. According to \cite{will}, a solution for supporting high data rates is to increase the spectral efficiency through advances in multiple-input multiple-output (MIMO) systems. 
Marzetta, in his seminal article \cite{noncoop}, introduced massive MIMO as a promising technology that significantly increases the spectral efficiency. In \cite{noncoop}, a perfectly synchronized massive MIMO system is considered and it is shown that the synchronous assumption is the worst possible case from the standpoint of the so-called pilot contamination phenomenon. However, as mentioned in \cite{three}, time-synchronous transmission is hard to attain over a large coverage area. 
In addition, the worst case, in terms of the pilot contamination, does not necessarily lead to the lowest ergodic rate, because the inter-cellular interference is negligible in the limit of an infinite number of antennas \cite{noncoop}, however for a large but finite number of antennas, the inter-cellular interference is also important.
Thus, it is  interesting to analyze an asynchronous massive MIMO system.

Asynchrony is addressed in \cite{three}, where an uplink massive MIMO system is analyzed. In \cite{three}, it is assumed that transmissions in each cell are synchronous, while pilot and uplink data transmissions in different cells are asynchronous. The analysis in \cite{three} indicates that the synchrony or asynchrony has no impact on the uplink transmission performance. 
In \cite{khan,DLcapacity,precoding}, the downlink direction of the synchronous massive MIMO systems is analyzed in terms of the achievable rate, the pilot contamination problem, and efficient precoding designs. The authors in \cite{khan} derive the achievable rate of the system for both maximum ratio combining (MRC) and zero forcing (ZF) precoders. 
In addition, \cite{DLcapacity} analyzes the downlink user capacity under the pilot-contaminated scenario. 
Finally, in \cite{precoding}, a new multi-cell minimum mean square error (MMSE) based precoding method is proposed that mitigates the pilot contamination problem.

Stochastic geometry is a powerful tool to evaluate the performance of large scale networks \cite{Haeng,Azimi1,Azimi2}. Here, it is assumed that the base stations are distributed randomly. The authors of \cite{tract} showed that the approach of using randomly distributed base stations is not only more tractable for system analysis but also as accurate as a grid model. In the literature, stochastic geometry has been rarely considered for the performance evaluation of synchronous massive MIMO systems \cite{emil,gil,heath,Madh,interference}. In these works, base stations are assumed to be distributed according to a homogeneous Poisson point process (HPPP) \cite{Haeng}. 
Particularly, \cite{emil} maximizes the uplink energy efficiency with respect to different system parameters. 
In \cite{gil}, stochastic geometery is used to develop a mathematical framework for computing the coverage probability and the ergodic rate. 
In addition to assuming an HPPP for the base stations, \cite{heath} models the distribution of the users with the same pilot sequence except for the desired user by an HPPP outside a ball centered at the desired base station location, i.e., an exclusion ball. In contrast, \cite{Madh} and \cite{interference} consider the coverage area of each cell as a circle around each base station, with possible overlap among the areas of adjacent cells.

The main contributions of our paper are as follows.

\textbf{Asynchronous and synchronous downlink massive MIMO modeling and analysis:} In this paper, we analyze a massive MIMO system, whose cells are not synchronous, while the transmissions in each cell are still synchronous. 
In order to study an asynchronous downlink massive MIMO system, we first estimate the channel coefficients and compute the downlink signal-to-interference-plus-noise ratio (SINR). Then, the coverage probability and the ergodic rate are derived by using stochastic geometry. 
Moreover, we analyze massive MIMO systems in the synchronous mode and compare the results with those achieved in the asynchronous mode. 
Our results are presented in the cases with an exclusion ball model as in \cite{heath}. We also consider fractional power control in the uplink transmission to compensate for a fraction of the path loss and to mitigate the near-far problem from intra-cell interference.

\textbf{Distribution of distances:} We derive the distributions of various distances, which play key roles in coverage probability and the ergodic rate analyses. In this context, we will acquire the probability distribution function (PDF) of three types of distances, namely, \textit{i)} the distance between a user and its serving base station, 
\textit{ii)} the distance between a user and its serving base station given the distance between the same user and another arbitrary base station, and
 \textit{iii)} the distance between a user and its serving base station given the distance between another user and its serving base station as well as the distance between these two users.

\textbf{System design insights:} Through simulation evaluations, we validate the analytical results and derive the downlink ergodic rate of a cell as a function of the uplink power control parameter and the number of pilot symbols. We observe that there are optimal values for these parameters maximizing the downlink ergodic rate of the cell. 

Our results indicate that using uplink full power control in the asynchronous mode leads to zero downlink rate. In addition, in most considered cases, we observe higher downlink ergodic rate in the synchronous case, compared to the cases with asynchronous transmission. Hence, the synchronous assumption is not necessarily the worst possible scenario, in terms of ergodic rate.

The paper is organized as follows. In Section \ref{sys_model_sec}, the system model is given. Section \ref{ch_est_sec} presents the channel estimation procedure. Section \ref{sto_sec} analyzes the downlink massive MIMO system using stochastic geometry and derives closed-form expressions for the coverage probability and the ergodic rate. Simulation results and discussions are outlined in Section \ref{sim_res}. Finally, Section \ref{conclude_sec} concludes the paper.

The following notations are used through the paper. Bold lower-case letters denote vectors, and bold upper-case letters represent matrices. $\mathbb{E}\{\cdot\}$, $\mathbb{P}\{\cdot\}$, $f\left(\cdot\right)$, $\delta\left(\cdot,\cdot\right)$, and $\mathbb{I}\left(\cdot\right)$ denote the expectation, the probability, the PDF, the Kronecker delta, and the indicator functions, respectively. In addition, the Euclidean norm is represented by $\norm{\cdot}$, and the absolute value is denoted by $\left|\cdot\right|$. We use $\mathbf{X}^*$, $\mathbf{X}^T$, $\mathbf{X}^H$ to represent the conjugate,  the transpose, and the Hermitian transpose of $\mathbf{X}$, respectively. $\mathbf{I}_M$ stands for the $M\times M$ identity matrix. We use $\mathbb{C}$ to represent the sets of all complex-valued numbers. Finally, we use $\mathcal{CN}\left(\cdot,\cdot\right)$ to denote a multi-variate circularly-symmetric complex Gaussian distribution.

\section{System Model}\label{sys_model_sec}
We consider a cellular network, operating under 6 GHz band, with one base station in each cell. Each base station has an antenna array with $M$ antennas simultaneously scheduling $K< M$ single-antenna users. Figure \ref{sys_model} shows the system model for $K=2$. 
\begin{figure}[!t]
\centering
\includegraphics[width=7cm,height=6cm]{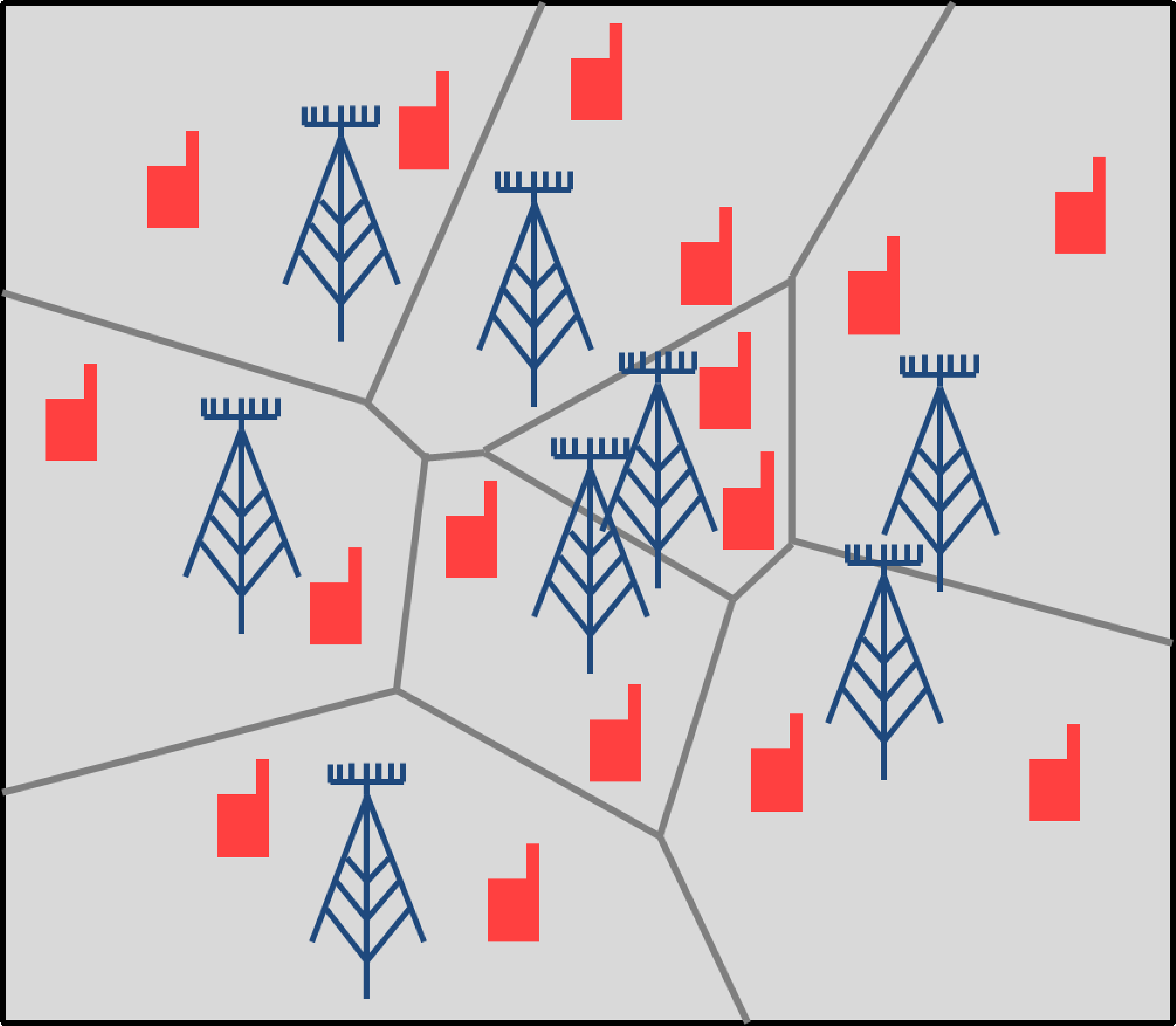}
\caption{The system model for $K=2$. Each base station has $M$ antennas, while users have one antenna.}\label{sys_model}
\end{figure}
The channel between each user and its serving base station is assumed to remain constant during a coherence time interval, denoted by $T_\text{c}$, which is equivalent to $N_\text{tot}$ symbols transmissions. In addition, the channel model is considered frequency-nonselective.
 
It is assumed that the system uses time-division duplexing (TDD) for transmission in the uplink and downlink directions. Thus, it is sufficient to estimate the channel vector in one direction. The channel vector is estimated using pilot sequences. In massive MIMO, it is usually assumed that users send the pilot sequences and the base stations estimate users' channel vectors. Different and orthogonal pilot sequences are assigned to users of each cell. The pilot sequence of the $k$-th user is denoted by $\boldsymbol{\varphi}_k\in \mathbb{C}^{N_\text{p}\times 1}$, each element with magnitude of one. Therefore, we have 
\begin{eqnarray}
\boldsymbol{\varphi}_k^H\boldsymbol{\varphi}_l=N_\text{p}\delta(l,k).
\end{eqnarray}
Due to the coherence time limitation, $N_\text{p}$ cannot be large. Hence, we consider the same set of orthogonal sequences in all cells, i.e., user $k$ of each cell, uses the same pilot sequence. Since the maximum number of mutually orthogonal sequences of the length $N_\text{p}$ is equal to $N_\text{p}$, we assume that $K=N_\text{p}$.

There are three transmission phases during a coherence time interval, namely, pilot transmission (channel estimation), downlink data transmission, and uplink data transmission (see Fig. \ref{sys_model}).
\begin{figure}[!t]
\centering
\includegraphics[width=6.5cm,height=1.6cm]{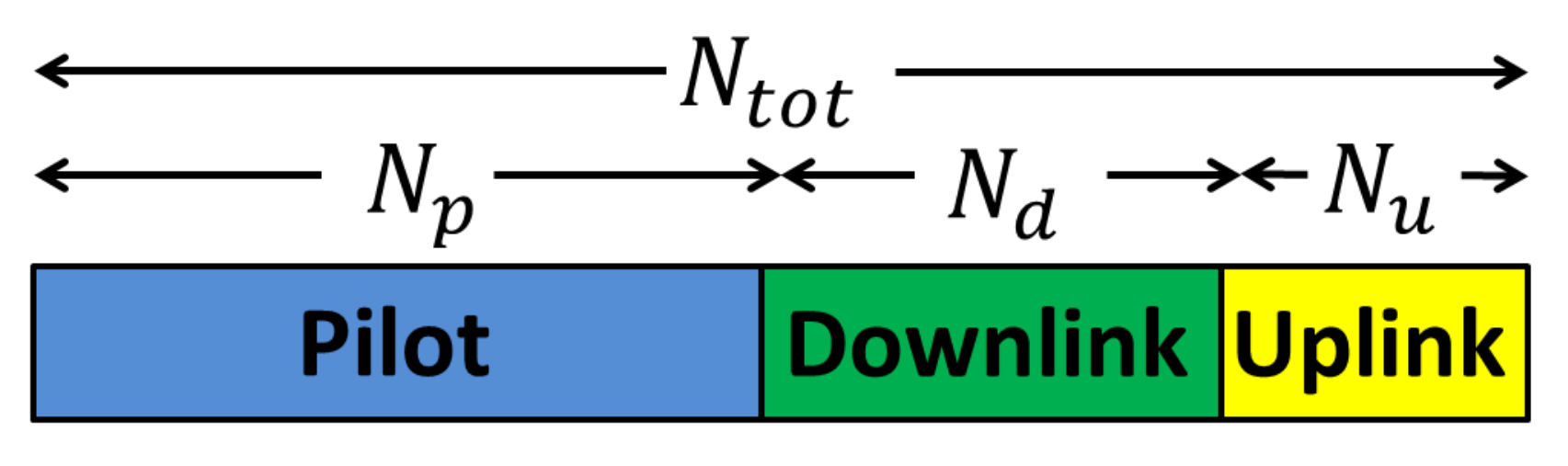}
\caption{The transmission protocol during a coherence time interval.}\label{sys_model}
\end{figure}
 In the first phase, the users transmit the pilot sequences. At the end of pilot transmission phase, the base station estimates the channel vectors of its serving users. Hence, the base station derives the precoding and combining vectors for downlink and uplink data transmission accordingly. Next, in the downlink data transmission, the base station sends the users' downlink data by using the precoding vector for each user derived using the user channel estimation vector. 
Finally, in the last phase, the users transmit their uplink data, and the base station detects the users' data with the help of the combining vectors. The number of downlink and uplink symbols during each coherence block are denoted by $N_\text{d}$ and $N_\text{u}$, respectively. Also, we define $Z=\frac{N_\text{d}}{N_\text{u}}$. In the literature, it is usually assumed that all cells are synchronous \cite{noncoop,khan,DLcapacity,precoding,emil,gil,heath,interference}. We refer to this case as synchronous mode. In contrast, while we assume synchronous transmission within each cell, the users in different cells are asynchronous. We refer to this case as asynchronous mode. In this paper, we analyze both the synchronous and the asynchronous modes.

\subsection{Signaling and Channel Model}
We consider both the small scale fading and the large scale path loss. The large scale path loss depends on the distance between the transmitter and the receivers, and the small scale fading is considered as a complex-Gaussian distributed random variable. 
It is notable that we do not consider the shadowing effect.

There are three types of channels:
\begin{itemize}
\item
\textbf{The channel between a user and a base station:} The channel vector between the $k$-th user of the $j$-th cell and the $l$-th base station is denoted by $\mathbf{h}_{ljk}\in \mathbb{C}^{M\times 1}$ and modeled by
\begin{eqnarray}
\label{wr1}
\mathbf{h}_{ljk}=\sqrt{\beta_{ljk}}\mathbf{g}_{ljk},\,
\beta_{ljk}=\omega r_{ljk}^{-\alpha},\,
\mathbf{g}_{ljk}\sim \mathcal{CN}\left(\zbm,\mathbf{I}_M\right),
\end{eqnarray}
where $\beta_{ljk}$ and $\mathbf{g}_{ljk}$ stand for the large scale path loss and the small scale fading, respectively. In (\ref{wr1}), $r_{ljk}$ denotes the distance between the $k$-th user of the $j$-th cell and $l$-th base station, $\alpha>2$ is the path loss exponent, and $\omega$ stands for the path loss at a reference distance of 1 km.

Since $M$ is large, according to the law of large numbers, we approximately have \cite{noncoop}
\begin{eqnarray}
\frac{1}{M}\mathbf{h}_{ljk}^H\mathbf{h}_{l'tk'}
=\frac{1}{M}\sqrt{\beta_{ljk}\beta_{l'tk'}}\mathbf{g}_{ljk}^H\mathbf{g}_{l'tk'}
\approx\beta_{ljk}\delta(l,l')\delta(j,t)\delta(k,k').
\end{eqnarray}
\item
\textbf{The channel between two base stations:} The channel matrix between the $j$-th base station (transmitter) and the $l$-th base station (receiver) is denoted by $\mathbf{H}_{lj}\in \mathbb{C}^{M\times M}$. Here, the channel is modeled as
\begin{eqnarray}
\mathbf{H}_{lj}=\sqrt{\beta_{lj}}\mathbf{G}_{lj},\quad\beta_{lj}=\omega r_{lj}^{-\alpha},
\end{eqnarray}
where $r_{lj}$ is the distance between the $l$-th base station and the $j$-th one. We let $\beta_{lj}$ and $\mathbf{G}_{lj}$ denote the large scale path loss and the small scale fading, respectively.
\item
\textbf{The channel between two users:} The channel between the $k$-th user of $l$-th cell and the $k'$-th user of the $j$-th cell is a scalar which is denoted by $h_{lkjk'}$ and modeled by
\begin{eqnarray}
h_{lkjk'}=\sqrt{\beta_{lkjk'}}g_{lkjk'},\,
\beta_{lkjk'}=\omega r_{lkjk'}^{-\alpha},\,
g_{lkjk'}\sim\mathcal{CN}(0,1),
\end{eqnarray}
where $r_{lkjk'}$ is the distance between the two users, and $\beta_{lkjk'}$ and $\mathbf{g}_{lkjk'}$ denote the large scale path loss and the small scale fading, respectively.
\end{itemize}

In the downlink phase, the $l$-th base station uses a precoding vector, $\mathbf{w}_{lk}$, in order to transmit data to the $k$-th user of its cell, where 
$\mathbf{w}_{lk}=\frac{\mathbf{u}_{llk}^*}{\norm{\mathbf{u}_{llk}}}.$ 
$\mathbf{u}_{llk}$ is the observation vector of $\mathbf{h}_{llk}$, obtained by multiplying 
the $l$-th base station received signal in pilot transmission phase 
by the $k$-th user's pilot sequence. Base stations transmit the downlink users' data with constant power $P_\text{d}$.

The $l$-th base station uses the combining vector $\overline{\mathbf{h}}_{llk}^H$ to detect the uplink data of the $k$-th user in its cell, where $\overline{\mathbf{h}}_{llk}$ is the linear minimum mean square error (LMMSE) channel estimation of $\mathbf{h}_{llk}$. In the uplink data and pilot transmission phases, users transmit their data using fractional power control, the same as used in the LTE \cite{heath}, \cite{LTE}, i.e., the $k$-th user of the $l$-th cell uses power $P_{lk}=P_\text{u}\beta_{llk}^{-\epsilon}$, where $0\leq\epsilon\leq1$ is the power control parameter, and $P_\text{u}$ is the uplink open loop transmit power. When $\epsilon =0$, there is no power control, and when $\epsilon = 1$, we have full power control.

The uplink and downlink signal of the $k$-th user of the $l$-th cell are denoted by $s_{lk}^\text{u}$ and $s_{lk}^\text{d}$, respectively. We consider $\mathbb{E}\left|s_{lk}^\text{u}\right|^2=\mathbb{E}\left|s_{lk}^\text{d}\right|^2=1$.
\subsection{Spatial Modeling}
Base stations are distributed by HPPP $\mathrm{\Phi}_\text{b}$ of density $\lambda$. Each base station simultaneously schedules $K$ users. It is assumed that a user connects to the nearest base station. Users of the same cell are distributed uniformly and independently over their Voronoi area, with the exclusion of a central disk of radius $r_0$ around their base station. Let's assume the $k$-th user of the $l$-th cell as the desired user.  Then, the $k$-th users of all cells except for the desired user are approximately distributed by the exclusion ball model introduced in \cite{tract}. According to the exclusion ball model, such users are distributed by the Poission point process (PPP) $\mathrm{\Phi}_{lk}^\text{u}$ of density 
$\lambda_1= \lambda\mathbb{I}(r>R_\text{e})$, 
which denotes that the interfering users are distributed by HPPP of density $\lambda$ outside a central disk of radius $R_\text{e}$ around the base station of the desired cell. As in \cite{heath}, we consider $R_\text{e}=\left(\pi\lambda\right)^{-\frac{1}{2}}$, whereby the average number of excluded users from the HPPP in the exclusion ball model is equal to 1.
\section{Channel Estimation}\label{ch_est_sec}
Channel vector is estimated using pilot signals. First, all users in the same cell simultaneously transmit their pilot sequences. Then, to estimate $\mathbf{h}_{llk}$, the $l$-th base station calculates the correlation between the received signal and $\boldsymbol{\varphi}_k$, the pilot sequence of the $k$-th user, by multiplying the received signal with $\frac{1}{N_\text{p}}\boldsymbol{\varphi}_k^*$, which leads to computing the observation vector, denoted by $\mathbf{u}_{llk}$. Finally, the $l$-th base station obtains LMMSE channel esimation vector, denoted by $\overline{\mathbf{h}}_{llk}$, based on $\mathbf{u}_{llk}$, as explained in the following.

In the asynchronous mode, while the $l$-th cell is in the pilot phase, other cells may be in one of the three phases.
 In order to consider these three cases, we define three binary random variables $\chi_{ljn}^\text{pp}$, $\chi_{ljn}^\text{pu}$, and $\chi_{ljn}^\text{pd}$ which can take the values $1$ and $0$. Consider that users of the $l$-th cell are transmitting their $n$-th symbols of the pilot sequence. Meanwhile, if the $j$-th cell is in the pilot phase, $\chi_{ljn}^\text{pp}=1$, but if the $j$-th cell is in the uplink phase, $\chi_{ljn}^\text{pu}=1$. Finally, if the $j$-th cell is in the downlink phase, we have $\chi_{ljn}^\text{pd}=1$. It is straight forward to show that $\chi_{ljn}^\text{pp}$, $\chi_{ljn}^\text{pu}$, and $\chi_{ljn}^\text{pd}$ take the value $1$ with probabilities $\frac{N_\text{p}}{N_\text{tot}}\frac{N_\text{p}}{N_\text{tot}}$, $\frac{N_\text{p}}{N_\text{tot}}\frac{N_\text{u}}{N_\text{tot}}$, and $\frac{N_\text{p}}{N_\text{tot}}\frac{N_\text{d}}{N_\text{tot}}$, respectively.

We assume that $S_l$ is the set of base stations that are synchronous with the $l$-th base station. In the asynchronous mode, we have $S_l=\{l\}$, and in the synchronous mode, $S_l$ includes all cells. In this way, the received signal  of the $l$-th base station in the pilot transmission phase can be expressed by
\begin{eqnarray}
\mathbf{Y}_l^\text{p}&=&\sum\limits_{j\in S_l}\sum\limits_{k'=1}^K \sqrt{P_{jk'}}\mathbf{h}_{ljk'}\left[\varphi_{k'}^{(1)}, ...,\varphi_{k'}^{(N_\text{p})}\right]
+\sum\limits_{j\not\in S_l}\sum\limits_{k'=1}^K \sqrt{P_\text{d}}\mathbf{H}_{lj}\mathbf{w}_{jk'}\left[\chi_{ljN_\text{p}}^\text{pd} s_{jk'}^{\text{d}(1)}, ...,\chi_{ljN_\text{p}}^\text{pd} s_{jk'}^{\text{d}(N_\text{p})}\right]\nonumber\\
&+&\sum\limits_{j\not\in S_l}\sum\limits_{k'=1}^K \sqrt{P_{jk'}}\mathbf{h}_{ljk'}
\biggl[\chi_{lj1}^\text{pp} s_{jk'}^{\text{p}(1)}+\chi_{lj1}^\text{pu} s_{jk'}^{\text{u}(1)}, ...,\chi_{ljN_\text{p}}^\text{pp} s_{jk'}^{\text{p}(N_\text{p})}+\chi_{lj1}^\text{pu} s_{jk'}^{\text{u}(N_\text{p})}\biggr]
+\mathbf{N},
\end{eqnarray}
where $\mathbf{N}\in \mathbb{C}^{M\times N_\text{p}}$ is the noise in the $l$-th base station (receiver), whose elements are independent and identically distributed complex Gaussian random variables with zero mean and variance $\sigma^2$.
Multiplying the received signal with $\frac{1}{N_\text{p}}\boldsymbol{\varphi}_k^*$, we have
\begin{eqnarray}
\mathbf{u}_{llk}&=&\sum_{j\in S_l}\sqrt{P_{jk}}\mathbf{h}_{ljk}
+\sum_{j\not\in S_l}\sum_{k'=1}^K \sqrt{P_{jk'}}\mathbf{h}_{ljk'}F\left(l,k,j,k'\right)
\nonumber\\
&+&\sum_{j\not\in S_l}\sum_{k'=1}^K \sqrt{P_\text{d}}\mathbf{H}_{lj}\mathbf{w}_{jk'}G\left(l,k,j,k'\right)
+\mathbf{N}\frac{\boldsymbol{\varphi}_k^*}{N_\text{p}},
\end{eqnarray}
where $F\left(l,k,j,k'\right)$ and $G\left(l,k,j,k'\right)$ are
\begin{eqnarray}
F\left(l,k,j,k'\right)&=&\left[\chi_{lj1}^\text{pp} s_{jk'}^{\text{p}(1)}, ...,\chi_{ljN_\text{p}}^\text{pp} s_{jk'}^{\text{p}(N_\text{p})}\right]\frac{\boldsymbol{\varphi}_k^*}{N_\text{p}}
+\left[\chi_{lj1}^\text{pu} s_{jk'}^{\text{u}(1)}, ...,\chi_{ljN_\text{p}}^\text{pu} s_{jk'}^{\text{u}(N_\text{p})}\right]\frac{\boldsymbol{\varphi}_k^*}{N_\text{p}},\\
G\left(l,k,j,k'\right)&=&\left[\chi_{lj1}^\text{pd} s_{jk'}^{\text{d}(1)}, ...,\chi_{ljN_\text{p}}^\text{pd} s_{jk'}^{\text{d}(N_\text{p})}\right]\frac{\boldsymbol{\varphi}_k^*}{N_\text{p}}.
\end{eqnarray}
It is straightforward to show that $F\left(l,k,j,k'\right)$ and $G\left(l,k,j,k'\right)$ have zero mean and the variances $\frac{N_\text{p}+N_\text{u}}{N_{\text{tot}}^2}$ and $\frac{N_\text{d}}{N_{\text{tot}}^2}$, respectively. 
Since $\mathbf{u}_{llk}$ depends on other users' channel estimation, it cannot be concluded that $\mathbf{u}_{llk}$ has necessarily Gaussian distribution. For instance, if at the time of downlink transmission of the $l$-th cell, the $j$-th cell starts transmitting pilot signals, the precoding vector of the $k'$-th user of the $j$-th cell $\mathbf{w}_{jk'}$  and the channel between the $l$-th base station and $j$-th base station $\mathbf{H}_{lj}$  will be dependent. However, we approximately consider Gaussian distribution for $\mathbf{u}_{llk}$. In Section \ref{sim_res}, we show that this is a good approximation. In this way, we have
\begin{eqnarray}
\mathbf{u}_{llk}&\sim& \mathcal{CN}\left(\zbm,\mathrm{\Delta}_{lk}\mathbf{I}_M\right),\\
\mathrm{\Delta}_{lk}&=&\frac{\mathbb{E}\left\|\mathbf{u}_{llk}\right\|^2}{M}
=\frac{P_{lk}}{M}\mathbb{E}\left\|\mathbf{h}_{llk}\right\|^2+\sum_{j\in S_l\backslash\{l\}}\frac{P_{jk}}{M}\mathbb{E}\left\|\mathbf{h}_{ljk}\right\|^2
\nonumber\\
&+&\sum_{j\not\in S_l}\sum_{k'=1}^K \left(\frac{P_{jk'}}{M}\mathbb{E}\left\|\mathbf{h}_{ljk'}\right\|^2\mathbb{E}\left|F\left(l,k,j,k'\right)\right|^2
+ \frac{P_\text{d}}{M}\mathbb{E}\left\|\mathbf{H}_{lj}\mathbf{w}_{jk'}\right\|^2\mathbb{E}\left|G\left(l,k,j,k'\right)\right|^2\right)
+\frac{\sigma^2}{N_\text{p}}\nonumber\\
&=&P_{lk}\beta_{llk}+\sum_{j\in S_l\backslash\{l\}}P_{jk}\beta_{ljk}
+\frac{N_\text{p}+N_\text{u}}{N_{\text{tot}}^2}\sum_{j\not\in S_l}\sum_{k'}P_{jk'}\beta_{ljk'}
+P_\text{d}\frac{N_\text{p}N_\text{d}}{N_{\text{tot}}^2}\sum_{j\not\in S_l}\beta_{lj}
+\frac{\sigma^2}{N_\text{p}}.
\end{eqnarray}
Using the observation vectors $\mathbf{u}_{llk'}$, the LMMSE channel estimation of $\mathbf{h}_{ljk}$ can be written as 
$\overline{\mathbf{h}}_{ljk}=\sum_{k'=1}^K A_{k'}\mathbf{u}_{llk'}$. The optimal coefficient of the estimator is obtained by applying the orthogonality principle as
\begin{eqnarray}
A_{k'}=\frac{\mathbb{E}\left\{\mathbf{h}_{ljk}^H\mathbf{u}_{llk'}\right\}}{\mathbb{E}\left|\mathbf{u}_{llk'}\right|^2}=\frac{\sqrt{P_{jk}}\beta_{ljk}}{\mathrm{\Delta}_{lk'}}\begin{cases}\delta\left(k,k'\right)&j\in S_l,\\ F\left(l,k',j,k\right)& j\not\in S_l.\end{cases}
\end{eqnarray}
Thus, we have
\begin{eqnarray}\label{ch_est}
\overline{\mathbf{h}}_{ljk}
=
\sqrt{P_{jk}}\beta_{ljk}\begin{cases}\frac{1}{\mathrm{\Delta}_{lk}}\mathbf{u}_{llk}&j\in S_l,\\ \sum_{k'=1}^K\frac{\mathbf{u}_{llk'}}{\mathrm{\Delta}_{lk'}}F\left(l,k',j,k\right)& j\not\in S_l.\end{cases}
\end{eqnarray}
Finally, the distribution of $\overline{\mathbf{h}}_{ljk}$ is obtained
\begin{eqnarray}
\label{ch_es}
\overline{\mathbf{h}}_{ljk} \sim
\begin{cases}
\mathcal{CN} \left (\zbm,\frac{P_{jk} \beta_{ljk}^2}{\mathrm{\Delta}_{lk}}\mathbf{I}_M\right )&j\in S_l,\\
\mathcal{CN} \left (\zbm,\frac{N_\text{p}+N_\text{u}}{N_{\text{tot}}^2}\sum_{k'=1}^K\frac{P_{jk} \beta_{ljk}^2}{\mathrm{\Delta}_{lk'}}\mathbf{I}_M\right )& j\not\in S_l.\end{cases}
\end{eqnarray}
Therefore, channel estimation error, defined as $\widetilde{\mathbf{h}}_{ljk}=\mathbf{h}_{ljk}-\overline{\mathbf{h}}_{ljk}$, has the following distribution
\begin{eqnarray}
\label{ch_err1}
\widetilde{\mathbf{h}}_{ljk} \sim \mathcal{CN} \left (\zbm,\beta_{ljk}\mathbf{I}_M-\mathbb{E}\left\{\overline{\mathbf{h}}_{ljk}\overline{\mathbf{h}}_{ljk}^H\right\}\right ),
\end{eqnarray}
where $\mathbb{E}\left\{\overline{\mathbf{h}}_{ljk}\overline{\mathbf{h}}_{ljk}^H\right\}$ can be obtained through (\ref{ch_es}) for $j\in S_l$ and $j\not\in S_l$.
\section{Stochastic Geometry Analysis}\label{sto_sec}
Consider that the $l$-th cell is transmitting its $n$-th downlink symbols. Meanwhile, the $k$-th user of the $l$-th cell can be exposed to different types of interfering signals, namely downlink, uplink, and pilot transmissions of other cells. Figure \ref{dlfig} shows all types of these interfering signals.
In this way, the received signal of the $k$-th user in the $l$-th cell is
\begin{eqnarray}
y_{lk}&=&
\sqrt{P_\text{d}}\sum\limits_{k'=1}^{K}\mathbf{h}_{llk}^T\mathbf{w}_{lk'}s_{lk'}^\text{d}
+\sqrt{P_\text{d}}\sum\limits_{j\in S_l\backslash \{l\}}\sum\limits_{k'=1}^K\mathbf{h}_{jlk}^T\mathbf{w}_{jk'}s_{jk'}^\text{d}
+\sum\limits_{j\not\in S_l}\sum\limits_{k'=1}^K\chi_{ljn}^\text{dd}\sqrt{P_\text{d}}\mathbf{h}_{jlk}^T\mathbf{w}_{jk'}s_{jk'}^\text{d}\nonumber\\
&+&\sum\limits_{j\not\in S_l}\sum\limits_{k'=1}^K\left(\chi_{ljn}^\text{dp}+\chi_{ljn}^\text{du}\right)\sqrt{P_{jk'}}h_{lkjk'}s_{jk'}^\text{d}
+n_\text{d},\, n_\text{d}\sim \mathcal{CN}\left(0,\sigma^2\right),
\end{eqnarray}
where $\chi_{ljn}^\text{dp}$, $\chi_{ljn}^\text{du}$, and $\chi_{ljn}^\text{dd}$ are binary random variables which can take the values $1$ and $0$. Consider that the $l$-th cell users are receiving their $n$-th symbols of the downlink data. 
Meanwhile, if the $j$-th cell's users transmit pilot sequences, $\chi_{ljn}^\text{dp}=1$, but if they transmit uplink signals, $\chi_{ljn}^\text{du}=1$. Finally, if the $j$-th cell transmits downlink signals, we have $\chi_{ljn}^\text{dd}=1$. It is straight forward to show that $\chi_{ljn}^\text{dp}$, $\chi_{ljn}^\text{du}$, and $\chi_{ljn}^\text{dd}$ take the value $1$ with probability $\frac{N_\text{d}}{N_\text{tot}}\frac{N_\text{p}}{N_\text{tot}}$, $\frac{N_\text{d}}{N_\text{tot}}\frac{N_\text{u}}{N_\text{tot}}$, and $\frac{N_\text{d}}{N_\text{tot}}\frac{N_\text{d}}{N_\text{tot}}$, respectively.

\begin{figure}[t]
\centering
\includegraphics[width=8cm,height=4.7cm]{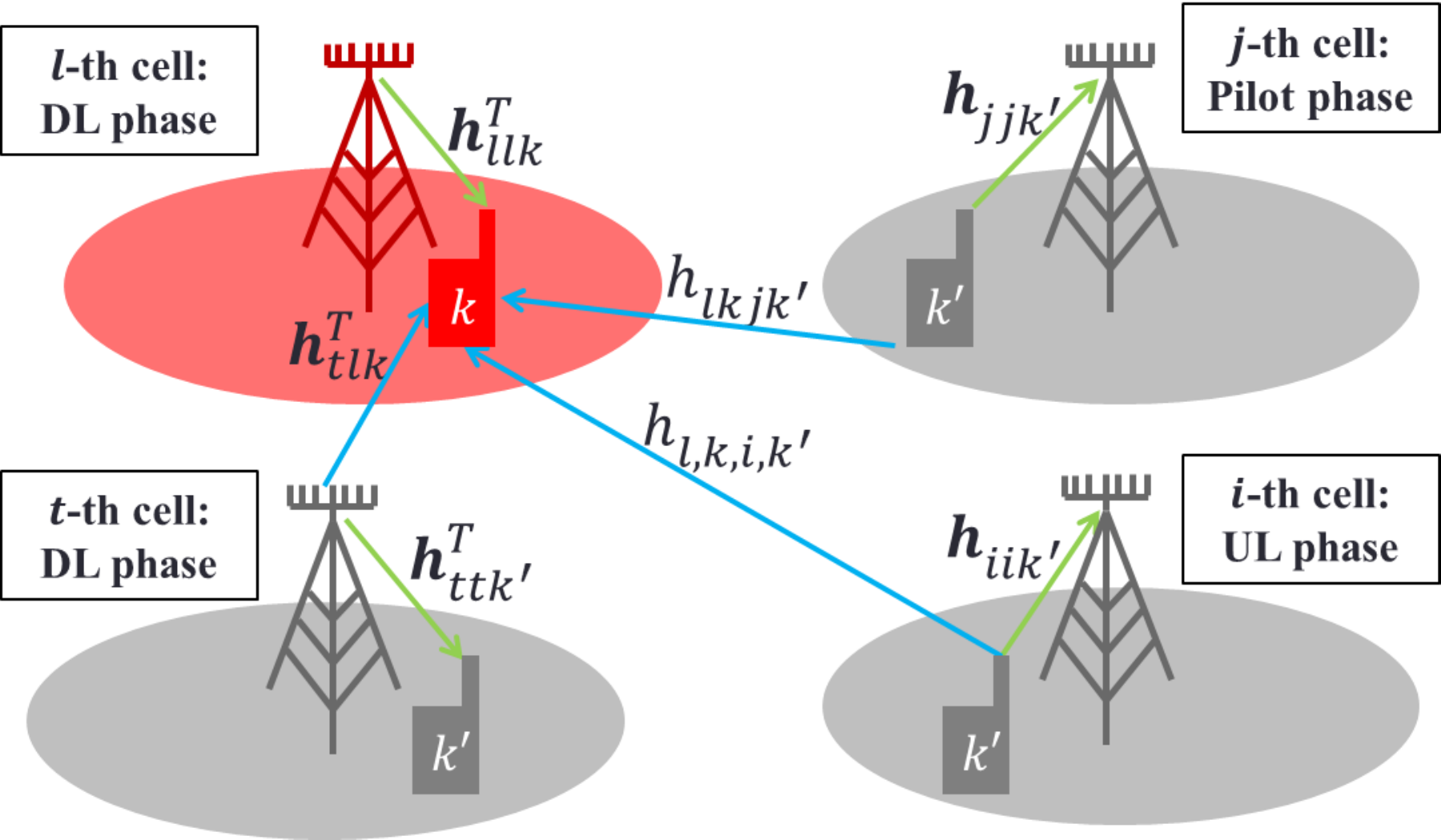}
\caption{Received signals of the desired user in the downlink phase of an asynchronous system. The desired cell is shown with red color. Green vectors indicate the useful signals and the bule ones show the interfering signals. Also, each link is labeled with its channel vector.}
\label{dlfig}
\end{figure}

It is assumed that the desired user, the $k$-th user of the $l$-th cell, knows its related value of $\mathbb{E}\{\mathbf{h}_{llk}^T\mathbf{w}_{lk}\}$. Hence, if we consider that the unkown part of the received signal as an uncorrelated noise, we obtain
%\begin{figure*}[!t] 
%\normalsize 
%\begingroup\makeatletter\def\f@size{9}\check@mathfonts
\begin{eqnarray}\label{sinr}
\mathrm{SINR}&=&\frac{\left(\mathbb{E}\left\{\mathbf{h}_{llk}^T\mathbf{w}_{lk}\right\}\right)^2}{I_\text{tot}},
\end{eqnarray}
where
\begin{eqnarray}
I_\text{tot}&=&
\mathrm{var}\left\{\mathbf{h}_{llk}^T\mathbf{w}_{lk}\right\}
+\sum\limits_{k'\ne k} \mathbb{E}\left|\mathbf{h}_{llk}^T\mathbf{w}_{lk'}\right|^2
+\sum\limits_{j\in S_l\backslash \{l\}}\sum\limits_{k'}\mathbb{E}\left|\mathbf{h}_{jlk}^T\mathbf{w}_{jk'}\right|^2\nonumber\\
&+&\sum\limits_{j\not\in S_l}\sum\limits_{k'}\chi_{ljn}^\text{dd}\mathbb{E}\left|\mathbf{h}_{jlk}^T\mathbf{w}_{jk'}\right|^2
+\sum\limits_{j\not\in S_l}\sum\limits_{k'}\dfrac{P_{jk'}}{P_\text{d}}\left(\chi_{ljn}^\text{dp}+\chi_{ljn}^\text{du}\right)\mathbb{E}\left|h_{lkjk'}\right|^2
+\dfrac{\sigma^2}{P_\text{d}}
.
\end{eqnarray}
%\endgroup
%\hrulefill 
%\vspace*{4pt} 
%\end{figure*}

By some calculations, the inverse of the signal-to-interference-plus-noise ratio \footnote{Since, in the process of obtaining the coverage probability, we need the inverse of the SINR, we derive $\mathrm{SINR}^{-1}$ instead of $\mathrm{SINR}$.}, i.e., $\mathrm{SINR}^{-1}$, is obtained as 
\normalsize 
\begingroup\makeatletter\def\f@size{11}\check@mathfonts
\begin{eqnarray}\label{sinr_1}
&&\mathrm{SINR}^{-1}=\gamma_1+\gamma_2+\gamma_3,\\
\label{gam1}
&&\gamma_1=\frac{V_M-1}{C_M^2}
+\frac{N_\text{p}}{C_M^2}
+\frac{\sigma^2r_{llk}^\alpha}{P_\text{d}\omega C_M^2}
+\frac{\sigma^2r_{llk}^{\alpha\left(1-\epsilon\right)}}{P_\text{u}C_M^2\omega^{1-\epsilon}}
+\frac{\sigma^4r_{llk}^{\alpha\left(2-\epsilon\right)}}{N_\text{p}P_\text{u}P_\text{d} C_M^2\omega^{2-\epsilon}}
+\frac{r_{llk}^{\alpha\left(1-\epsilon\right)}}{C_M^2}\left(N_\text{p}+\frac{\sigma^2}{P_\text{d}\omega}r_{llk}^{\alpha}\right)\times\nonumber\\
&&\qquad\left(\sum_{j\in S_l\backslash\{l\}}r_{jjk}^{\alpha\epsilon}r_{ljk}^{-\alpha}
+\frac{N_\text{p}+N_\text{u}}{N_{\text{tot}}^2}\sum_{j\not\in S_l}\sum_{k'}r_{jjk'}^{\alpha\epsilon}r_{ljk'}^{-\alpha}
+\frac{P_\text{d}N_\text{p}N_\text{d}}{P_\text{u}\omega^{-\epsilon}N_{\text{tot}}^2}\sum_{j\not\in S_l}r_{lj}^{-\alpha}\right),\\
\label{gam2}
&&\gamma_2=\frac{N_\text{p}}{C_M^2}\left(r_{llk}^\alpha+r_{llk}^{\alpha \left(2-\epsilon\right)}\mathrm{\Delta}_{lk}^{(1)}\right)\sum_{j\in S_l\backslash \{l\}} r_{jlk}^{-\alpha}
+\frac{N_\text{p}}{C_M^2}\left(r_{llk}^\alpha+r_{llk}^{\alpha \left(2-\epsilon\right)}\mathrm{\Delta}_{lk}^{(1)}\right)\sum_{j\not\in S_l}\chi_{ljn}^\text{dd}r_{jlk}^{-\alpha}\nonumber\\
&&+\frac{M-1}{C_M^2}r_{llk}^{2\alpha}\sum_{j\in S_l\backslash \{l\}} \frac{\mathrm{\Delta}_{lk}}{\mathrm{\Delta}_{jk}}r_{jlk}^{-2\alpha}+\frac{M-1}{C_M^2}r_{llk}^{2\alpha}\frac{N_\text{p}+N_\text{u}}{N_{\text{tot}}^2}\sum_{j\not\in S_l}\chi_{ljn}^\text{dd}\sum\limits_{k'}\frac{\mathrm{\Delta}_{lk}}{\mathrm{\Delta}_{jk'}}r_{jlk}^{-2\alpha},\\
\label{gam3}
&&\gamma_3=\left(r_{llk}^{\alpha}+r_{llk}^{\alpha\left(2-\epsilon\right)}\mathrm{\Delta}_{lk}^{(1)}\right)\frac{P_\text{u}}{P_\text{d}\omega^{\epsilon}C_M^2}
\sum_{j\not\in S_l}\sum_{k'}\left(\chi_{ljn}^\text{dp}+\chi_{ljn}^\text{du}\right)r_{jjk'}^{\alpha\epsilon}r_{lkjk'}^{-\alpha},
\end{eqnarray}
\endgroup
where 
$\mathrm{\Delta}_{lk}^{(1)}=P_\text{u}^{-1}\omega^{\epsilon-1}\mathrm{\Delta}_{lk}-x^{-\alpha\left(1-\epsilon\right)}$
 (See Appendix \ref{sinrcal} for details). The terms in (\ref{gam1}) are caused by the intra-cellular interference and the noise. The terms in (\ref{gam2}) are from pilot contamination, and the expression (\ref{gam3}) is caused by inter-cellular interference in the uplink direcion.

To analyze the system performance, we need to find the PDF of the received SINR. For this reason, in Subsection \ref{dist_dist}, we first derive the distributions of different terms of distance. Then, using these distributions, we derive the PDF of the SINR and the coverage probability (Subsection \ref{cov_prob}) as well as the ergodic rate (Subsection \ref{rate_sub}).

\subsection{Distance Distribution}\label{dist_dist}
In this Subsection, we derive the distributions of the distances required for the coverage probability and the ergodic rate analysis.

\begin{theorem}\label{thm1}
The PDF of the distance between a user and its serving base station, i.e., $r$, is given by
\begin{eqnarray}\label{thm1_eq}
 f\left(r\right)=2\pi\lambda r e^{-\pi \lambda \left(r^2-r_0^2\right)},\quad r>r_0.\end{eqnarray}
\end{theorem}
\begin{proof}
The theorem is proved for $r_0=0$ in \cite{tract}. See Apendix \ref{pthm1} for $r_0\ne 0$.
\end{proof}
The expression of $\mathrm{SINR}^{-1}$ in (\ref{sinr_1}) has terms such as $r_{jjk'}^{\alpha\epsilon}r_{ljk'}^{-\alpha}$ and $r_{jjk'}^{\alpha\epsilon}r_{lkjk'}^{-\alpha}$, which have parameters correlated with each other. Hence, as will be presented in the following subsection, we need the PDF of $r_{jjk'}$ given $r_{ljk'}$ as well as the PDF of $r_{jjk'}$ given $r_{lkjk'}$ and $r_{llk}$, which are presented in Theorems \ref{thm2} and \ref{thm3}, respectively.

\begin{theorem}\label{thm2}
The PDF of the distance between a user and its serving base station, $r_1$, conditioned on the distance between this user and another arbitrary base station, $r_2$, is
\begin{eqnarray}f\left(r_1|r_2\right)=\frac{2\pi \lambda r_1 e^{-\pi \lambda r_1^2}}{e^{-\pi \lambda r_0^2}-e^{-\pi \lambda r_2^2}},\quad r_0<r_1<r_2.\end{eqnarray}
\end{theorem}
\begin{proof}
See Appendix \ref{pthm2}.
\end{proof}
\begin{theorem}\label{thm3}
The PDF of the distance between a user and its serving base station, $r_{jjk'}$, conditioned on the distance between another user and the user serving base station, $r_{llk}$, as well as the distance between these two users, $r_{lkjk'}$, is
\begin{eqnarray}
\mathbb{P}\biggl(r_{jjk'}=s\bigg|r_{lkjk'}=r,\,r_{llk}=x\biggr)
=\frac{2\pi\lambda s e^{-\pi\lambda s^2}}{e^{-\pi\lambda R_1^2}-e^{-\pi\lambda R_2^2}},\quad R_1<s<R_2,
\end{eqnarray}
where $R_1 = \max{\left(r_0,\,x-r\right)}$ and $R_2 = x+r$.
\end{theorem}
\begin{proof}
See Appendix \ref{pthm3}.
\end{proof}
\subsection{Coverage Probability}\label{cov_prob}
The coverage probability can be expressed as
\begin{eqnarray}\label{convert_org}
\mathbb{P}\left(\mathrm{SINR}>T\right)=\int \mathbb{P}\left(\mathrm{SINR}>T|r_{llk}=x\right) f\left(x\right) \,\mathrm{d}x,
\end{eqnarray}
where $f\left(\cdot\right)$ is given in Theorem \ref{thm1}. Due to the nature of massive MIMO systems, in which SINR depends on the large scale fading, the coverage probability can not be derived using previously known stochastic geometry procedures such as \cite{tract}. Thus, we need to use approximation schemes as follows. 
The conditional coverage probability in (\ref{convert_org}) can be approximated as
\begin{eqnarray}\label{cond1}
\mathbb{P}\left(\mathrm{SINR}>T|r_{llk}=x\right)
&=&\mathbb{P}\left(1>T\left(\mathrm{SINR}\right)^{-1}|r_{llk}=x\right)
\overset{(a)}{\approx}  \mathbb{P}\left(g>T\left(\mathrm{SINR}_x\right)^{-1}\right)\nonumber\\
&\overset{(b)}{\approx}&\sum_{n=1}^N \left(-1\right)^{n+1}\binom{N}{n}\mathbb{E}\left\{e^{-\eta nT\mathrm{SINR}_x^{-1}}\right\},
\end{eqnarray}
where $\mathrm{SINR}_x$ stands for the SINR under the condition of $r_{llk}=x$. Also, $(a)$ is according to \cite{gamma_limit}, in which $1$ is approximated by a dummy Gamma random variable with unit mean and the shape parameter of $N$, such that $\lim_{N\to\infty}{N^Ng^{N-1}e^{-Ng}}/{\Gamma\left(N\right)}=\delta\left(g-1\right)$, where $\delta(\cdot)$ is the Dirac delta function, and $\Gamma(\cdot)$ is the Gamma function defined as $\Gamma(a) =\int_0^\infty e^{-t}t^{a-1}\,\mathrm{d}t$. In Section \ref{sim_res}, we observe that it does not need to choose a very large $N$ that our analytical results converge to Monte Carlo results. Finally, $(b)$ is shown in Appendix \ref{gam_appr}.

Due to the complexity of $\mathrm{SINR}_x^{-1}$, we use an approximation of it, denoted by $\widehat{\mathrm{SINR}_x}^{-1}$, which is obtained by using the following approximations of the terms in (\ref{sinr_1}).
\begin{eqnarray}
\label{appr1}
&&\chi_{ljn}^\text{dd}\approx\mathbb{E}\left\{\chi_{ljn}^\text{dd}\right\}=\frac{N_\text{d}^2}{N_\text{tot}^2},\\
&&\chi_{ljn}^\text{dp}+\chi_{ljn}^\text{du}\approx\mathbb{E}\left\{\chi_{ljn}^\text{dp}+\chi_{ljn}^\text{du}\right\}=\frac{N_\text{d}\left(N_\text{p}+N_\text{u}\right)}{N_\text{tot}^2},\\
&&\mathrm{\Delta}_{lk}^{(1)}\approx Q_1= \mathbb{E}\left\{\mathrm{\Delta}_{lk}^{(1)}|r_{llk}=x\right\},\\
&&\frac{\mathrm{\Delta}_{lk}}{\mathrm{\Delta}_{jk}}\approx1,\\
\label{Q_2}
&&\sum_{j\not\in S_l}r_{lj}^{-\alpha}\approx Q_2=\mathbb{E}\left\{\sum_{j\not\in S_l}r_{lj}^{-\alpha}|r_{llk}=x\right\},\\
\label{appr2}
&&\sum_{j\not\in S_l}\sum_{k'}r_{jjk'}^{\alpha\epsilon}r_{lkjk'}^{-\alpha}\approx Q_3=\mathbb{E}\left\{\sum_{j\not\in S_l}\sum_{k'}r_{jjk'}^{\alpha\epsilon}r_{lkjk'}^{-\alpha}|r_{llk}=x\right\}.
\end{eqnarray}
Therefore, we have
\begin{eqnarray}
\widehat{\mathrm{SINR}}_x^{-1}=c_1(x)+e_1(x)+e_2(x),
\end{eqnarray}
where $c_1(x)$, $e_1(x)$, and $e_2(x)$ are given by
\begin{eqnarray}
\label{c_1}
c_1(x)&=&\frac{V_M-1}{C_M^2}
+\frac{N_\text{p}}{C_M^2}
+\frac{\sigma^2}{P_\text{d}\omega C_M^2}x^\alpha
+\frac{\sigma^2}{P_\text{u}C_M^2\omega^{1-\epsilon}}x^{\alpha\left(1-\epsilon\right)}
+\frac{\sigma^4}{N_\text{p}P_\text{u}P_\text{d} C_M^2\omega^{2-\epsilon}}x^{\alpha\left(2-\epsilon\right)}
\nonumber\\
&+&\left(x^{\alpha}+x^{\alpha\left(2-\epsilon\right)}Q_1\right)
\frac{P_\text{u}N_\text{d}\left(N_\text{p}+N_\text{u}\right)}{P_\text{d}\omega^{\epsilon}C_M^2N_\text{tot}^2}Q_3
+\left(N_\text{p}+\frac{\sigma^2}{P_\text{d}\omega}x^{\alpha}\right)
\frac{P_\text{d}N_\text{p}N_\text{d}x^{\alpha\left(1-\epsilon\right)}}{P_\text{u}\omega^{-\epsilon}C_M^2N_{\text{tot}}^2}
Q_2,\\
\label{e_1}
e_1(x)&=&\frac{N_\text{p}}{C_M^2}\left(x^\alpha+x^{\alpha \left(2-\epsilon\right)}Q_1\right)
\sum_{j\in S_l\backslash \{l\}} r_{jlk}^{-\alpha}+
\frac{N_\text{p}}{C_M^2}\left(x^\alpha+x^{\alpha \left(2-\epsilon\right)}Q_1\right)
\frac{N_\text{d}^2}{N_\text{tot}^2}\sum_{j\not\in S_l}r_{jlk}^{-\alpha}
\nonumber\\
&+&\frac{M-1}{C_M^2}x^{2\alpha}
\sum_{j\in S_l\backslash \{l\}} r_{jlk}^{-2\alpha}
+\frac{M-1}{C_M^2}x^{2\alpha}
\frac{N_\text{p}N_\text{d}^2\left(N_\text{p}+N_\text{u}\right)}{N_{\text{tot}}^4}
\sum_{j\not\in S_l}r_{jlk}^{-2\alpha},
\end{eqnarray}
and
\begin{eqnarray}
\label{e_2}
&&e_2(x)=\frac{x^{\alpha\left(1-\epsilon\right)}}{C_M^2}\left(N_\text{p}+\frac{\sigma^2}{P_\text{d}\omega}x^{\alpha}\right)
\left(\sum_{j\in S_l\backslash\{l\}}r_{jjk}^{\alpha\epsilon}r_{ljk}^{-\alpha}
+\frac{N_\text{p}+N_\text{u}}{N_{\text{tot}}^2}\sum_{j\not\in S_l}\sum_{k'}r_{jjk'}^{\alpha\epsilon}r_{ljk'}^{-\alpha}
\right).
\end{eqnarray}
In this way, using $\widehat{\mathrm{SINR}}_x^{-1}$ in (\ref{cond1}), the conditional coverage probability is approximated as
\begin{eqnarray}\label{prop_cond}
\mathbb{P}\left(\mathrm{SINR}>T|r_{llk}=x\right)
\approx\sum_{n=1}^N \left(-1\right)^{n+1}\binom{N}{n}
\mathcal{C}_1\left(T,n,x\right)\mathcal{E}_1\left(T,n,x\right)\mathcal{E}_2\left(T,n,x\right),
\end{eqnarray}
where
\begin{eqnarray}
\mathcal{C}_1\left(T,n,x\right)&=&\exp\left(-\eta nTc_1(x)\right),\\
\mathcal{E}_1\left(T,n,x\right)&=&\mathbb{E}_{\left\{r_{jlk}|j\ne l\right\}}\left\{\exp\left(-\eta nTe_1(x)\right)\right\},\\
\mathcal{E}_2\left(T,n,x\right)&=&\mathbb{E}_{\left\{r_{jjk'}r_{ljk'}|j\ne l\right\}}\left\{\exp\left(-\eta nTe_2(x)\right)\right\}.
\end{eqnarray}

Finally, by using (\ref{convert_org}), (\ref{prop_cond}), and Theorem \ref{thm1}, the coverage probability is obtained from its conditional form as
\begin{eqnarray}\label{cov_prob_frml}
&\mathbb{P}\left(\mathrm{SINR}>T\right)&\approx
\sum_{n=1}^N \left(-1\right)^{n+1}\binom{N}{n}\times\nonumber\\
&&\int_{r_0}^\infty e^{-\pi \lambda \left(x^2-r_0^2\right)}
\mathcal{C}_1\left(T,n,x\right) \mathcal{E}_1\left(T,n,x\right) \mathcal{E}_2\left(T,n,x\right)2\pi\lambda x\,\mathrm{d}x.
\end{eqnarray}

In Appendix \ref{cov_cal}, $\mathcal{C}_1\left(T,n,x\right)$, $\mathcal{E}_1\left(T,n,x\right)$, and $\mathcal{E}_2\left(T,n,x\right)$ in (\ref{cov_prob_frml}) are obtained for both asynchronous and synchronous modes.

\textbf{Full power control case ($\epsilon=1$) in the asynchronous mode:} Using uplink power control in the asynchronous mode can effectively increases the inter-cellular interference in the asynchronous mode. For $\epsilon=1$ and finite number of antennas, the inter-cellular interference which is caused by other cell's users is the dominant source of interference. Hence, we have 
\begin{eqnarray}
\widehat{\mathrm{SINR}}_x^{-1}\approx \left(x^{\alpha}+x^{\alpha\left(2-\epsilon\right)}Q_1\right)\frac{P_\text{u}}{P_\text{d}\omega^{\epsilon}C_M^2}
\frac{N_\text{d}\left(N_\text{p}+N_\text{u}\right)}{N_\text{tot}^2}Q_3,
\end{eqnarray}
and the coverage probability of the asynchronous mode for $\epsilon=1$ is
\begin{eqnarray}
&&\mathbb{P}\left(\mathrm{SINR}>T\right)\approx
\sum_{n=1}^N \left(-1\right)^{n+1}\binom{N}{n}\times\nonumber\\
&&\int_{r_0}^\infty
\exp\left(-\pi \lambda \left(x^2-r_0^2\right)-\eta nT
\left(x^{\alpha}+x^{\alpha\left(2-\epsilon\right)}Q_1\right)
\frac{P_\text{u}N_\text{d}\left(N_\text{p}+N_\text{u}\right)}{P_\text{d}\omega^{\epsilon}C_M^2N_\text{tot}^2}Q_3
\right)
2\pi\lambda x\,\mathrm{d}x.
\end{eqnarray}

\textbf{Infinite number of antennas:} In (\ref{c_1})-(\ref{e_2}), we observe that only the expression in the second term of (\ref{e_1}), which is caused by the pilot contamination problem, grows as the number of antennas $M$ increases. Hence, when $M$ tends to infinity, the coverage probability expression in (\ref{cov_prob_frml}) is simplified as
\begin{eqnarray}
&&\mathbb{P}\left(\mathrm{SINR}>T\right)\overset{M\to \infty}{\approx}
\sum_{n=1}^N \left(-1\right)^{n+1}\binom{N}{n}\times\nonumber\\
&&\int_{r_0}^\infty 
\exp\left(-\pi \lambda \left(x^2-r_0^2\right)+\int_x^\infty \left(\exp\left(C\left(T,n,x\right)r^{-2\alpha}\right)-1\right)2\pi\lambda r\,\mathrm{d}r\right)
2\pi\lambda x\,\mathrm{d}x,
\end{eqnarray}
where $C(x)$ is defined in (\ref{c_asyn}) and (\ref{c_syn}) for the asynchronous and synchronous modes, respectively.

\textbf{No power control case ($\epsilon=0$):} When there is no uplink power control, the experssion of $\mathcal{E}_2\left(T,n,x\right)$ can be expressed by a 1-D integral. Thus, the coverage probability in the asynchronous mode is obtained as
\begin{eqnarray}
&&\mathbb{P}\left(\mathrm{SINR}>T\right)\approx
\sum_{n=1}^N \left(-1\right)^{n+1}\binom{N}{n}\times\nonumber\\
&&\int_{r_0}^\infty\exp\biggl(-\pi \lambda \left(x^2-r_0^2\right)-\eta nTc_1(x)+
N_\text{p}
\int_{\pi\lambda R_\text{e}^2}^\infty 
\left[\exp\left(D^\text{asyn}\left(T,n,x\right)R_\text{e}^{-\alpha} r^{-\frac{\alpha}{2}}\right)-1\right]
\,\mathrm{d}r\nonumber\\
&&+\int_{\pi\lambda x^2}^\infty \left[\exp\left(
B^\text{asyn}\left(T,n,x\right)R_\text{e}^{-\alpha}t^\frac{-\alpha}{2}+C^\text{asyn}\left(T,n,x\right)R_\text{e}^{-2\alpha}t^{-\alpha}
\right)-1\right]\,\mathrm{d}t
\biggr)2\pi\lambda x\,\mathrm{d}x,
\end{eqnarray}
and in the synchronous mode, we have
\begin{eqnarray}
&&\mathbb{P}\left(\mathrm{SINR}>T\right)\approx
\sum_{n=1}^N \left(-1\right)^{n+1}\binom{N}{n}\times\nonumber\\
&&\int_{r_0}^\infty\exp\biggl(-\pi \lambda \left(x^2-r_0^2\right)-\eta nTc_1(x)+
\int_{\pi\lambda R_\text{e}^2}^\infty 
\left[\exp\left(D^\text{syn}\left(T,n,x\right)R_\text{e}^{-\alpha} t^{-\frac{\alpha}{2}}\right)-1\right]
\,\mathrm{d}t\nonumber\\
&&+\int_{\pi\lambda x^2}^\infty \left[\exp\left(
B^\text{syn}\left(T,n,x\right)R_\text{e}^{-\alpha}t^\frac{-\alpha}{2}+C^\text{syn}\left(T,n,x\right)R_\text{e}^{-2\alpha}t^{-\alpha}
\right)-1\right]\,\mathrm{d}t
\biggr)2\pi\lambda x\,\mathrm{d}x.
\end{eqnarray}
\subsection{Downlink Ergodic Rate}\label{rate_sub}
Since the coverage probability is a good metric for delay-sensitive applications but ergodic rate is good for delay-insensitive applications, we analyze the total ergodic rate of a cell. The total ergodic rate over all users of a cell is obtained as
\begin{eqnarray}\label{adr}
&&\mathcal{R}=\frac{N_\text{p}N_\text{d}}{N_\text{tot}}\mathbb{E}\left\{\log_2{\left(1+\mathrm{SINR}\right)}\right\}\overset{(c)}{=}\frac{N_\text{p}N_\text{d}}{N_\text{tot}}
\int_{s>0} \mathbb{P}\left(\log_2{\left(1+\mathrm{SINR}\right)}>s\right)\,\mathrm{d}s\nonumber\\
&&\overset{(d)}{=}\frac{N_\text{p}N_\text{d}}{N_\text{tot}}
\int_{t>0}\frac{ \mathbb{P}\left(\mathrm{SINR}>t\right)}{\left(t+1\right)\ln{2}}\,\mathrm{d}t
\overset{(e)}{\approx}\frac{N_\text{p}N_\text{d}}{N_\text{tot}}\sum_{n=1}^N \frac{\left(-1\right)^{n+1}}{\ln{2}}\binom{N}{n}\times\nonumber\\
&&\qquad\qquad\int_{r_0}^\infty 2\pi\lambda x e^{-\pi\lambda\left(x^2-r_0^2\right)}
\int_{t>0}\frac{ \mathcal{C}_1\left(t,n,x\right)\mathcal{E}_1\left(t,n,x\right)\mathcal{E}_2\left(t,n,x\right)}{\left(t+1\right)}\,\mathrm{d}t\,\mathrm{d}x,
\end{eqnarray}
where $(c)$ is obtained due to the fact that for a nonnegative random variable $S$ we have $\mathbb{E}\{S\}=\int_0^\infty\mathbb{P}\left(S>s\right)\,\mathrm{d}s$. In $(d)$, we change the variable as $t=2^s-1$ and in $(e)$, the expression in (\ref{cov_prob_frml}) is replaced. Note that in (\ref{adr}), $N_\text{p}$ is the number of users in a cell and $\frac{N_\text{d}}{N_\text{tot}}$ is the fraction of the coherence time which is dedicated to the downlink transmission.
\section{Simulation Results}\label{sim_res}
In this section, the analytical results of Section \ref{sto_sec} are validated by comparing them with Monte Carlo simulations. We evaluate the ergodic rate as a function of the number of pilot symbols and  the uplink power control parameter. We consider that the base stations are distributed by the Poisson point process of the density $\lambda$ in a square region whose sides have length $4000\,\textrm{m}$. Then, we generate users' points with the Poisson point process of a large density to make sure that there are $N_\text{p}$ users in each cell. Users connect to the nearest base station. In addition, users whose distance from the serving base station are less than $r_0$ are removed. Then, in each cell, $N_\text{p}$ users are selected randomly and the reminder of them are removed. In Table \ref{param}, the considered system parameters, which are in harmony with \cite{emil}, \cite{heath}, \cite{table_ref}, are presented.

\begin{table}[!t]
\centering
\caption{Simulation Parameters}\label{TB1}
\begin{tabular}{|c|c|}
\hline
\textbf{System Parameter}&\textbf{Value}\\
\hline
$P_\text{d}$ & $45\mathrm{dBm}$\\
$P_\text{u}$ & $23\mathrm{dBm}$\\
$\sigma_n^2$&$-200\mathrm{dBm}$\\
$\omega$&$130\mathrm{dB}$\\
$\alpha$&$4$\\
$M$&$[64,128,10000]$\\
$N_\text{tot}$&$40$\\
$N_\text{p}$&$10$\\
$Z$&$2$\\
$R_e$&$500m$\\
$r_0$&$50m$\\
\hline
\end{tabular}
\label{param}
\end{table}

In Fig. \ref{asyncov}, comparisons between the analytical results and Monte Carlo simulations of the asynchronous mode for different values of the power control parameter, $\epsilon$, are demonstrated. The results for the synchronous mode are also shown in Fig. \ref{syncov}. These figures show the coverage probability as a function of the threshold $T$. Note that the analytical results of (\ref{cov_prob_frml}) converge to the Monte Carlo simulations with a small value of the shape parameter of the Gamma random variable $N$. 
For example, in the asynchronous mode, for all values of $\epsilon$, $N$ in (\ref{cov_prob_frml}) equals $1$. However, in the synchronous mode, we consider $N=[2,3,4,8]$. 
As observed, the analytical results tightly mimic the exact Monte Carlo results for both synchronous and asynchronous systems and for different parameter settings.
\begin{figure}[t]
\centering
\includegraphics[width=9cm,height=6cm]{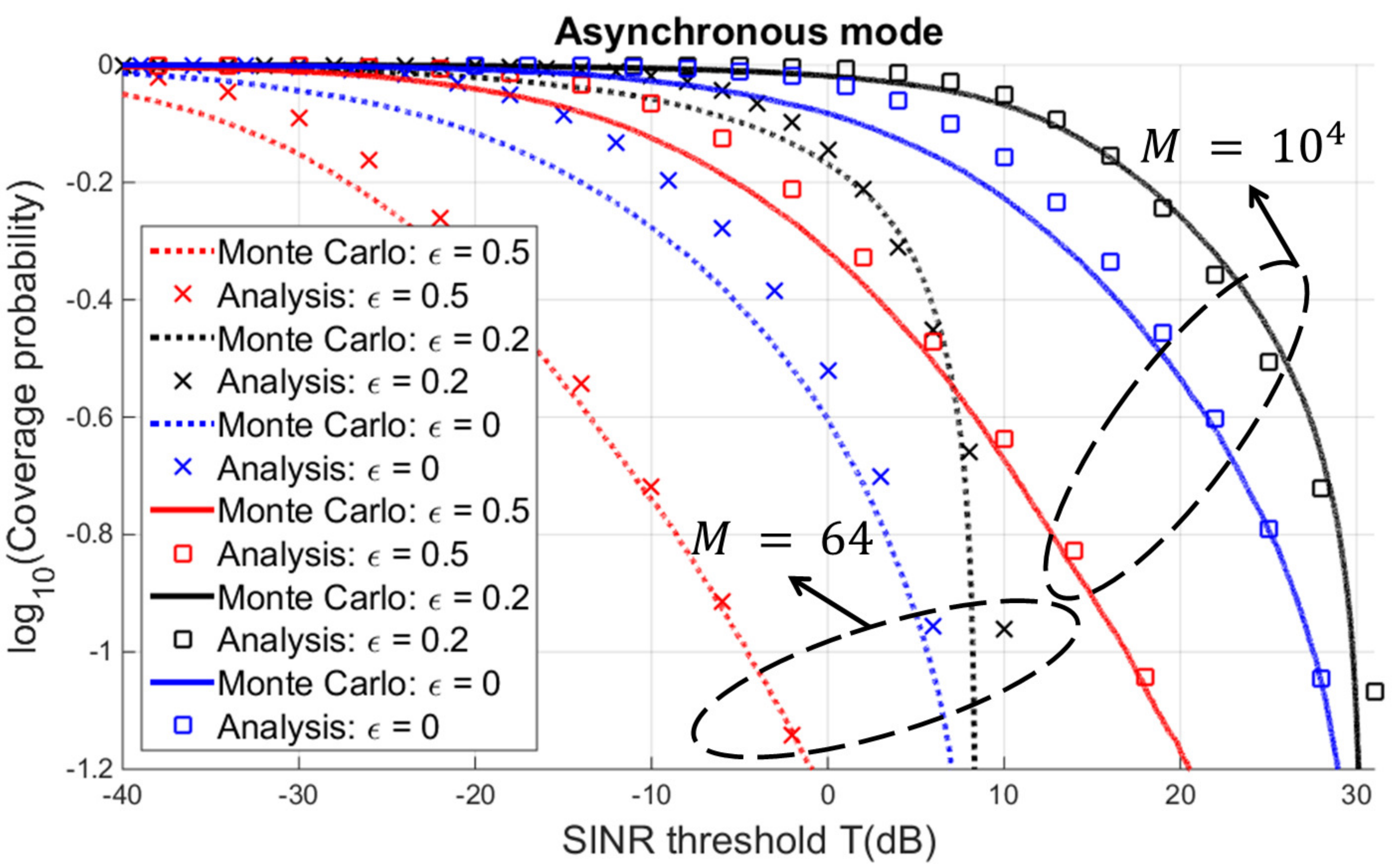}
\caption{Comparison between Monte Carlo and analytical results for different values of ($M$,$\epsilon$) in the asynchronous mode.}
\label{asyncov}
\end{figure}

\begin{figure}[t]
\centering
\includegraphics[width=9cm,height=6cm]{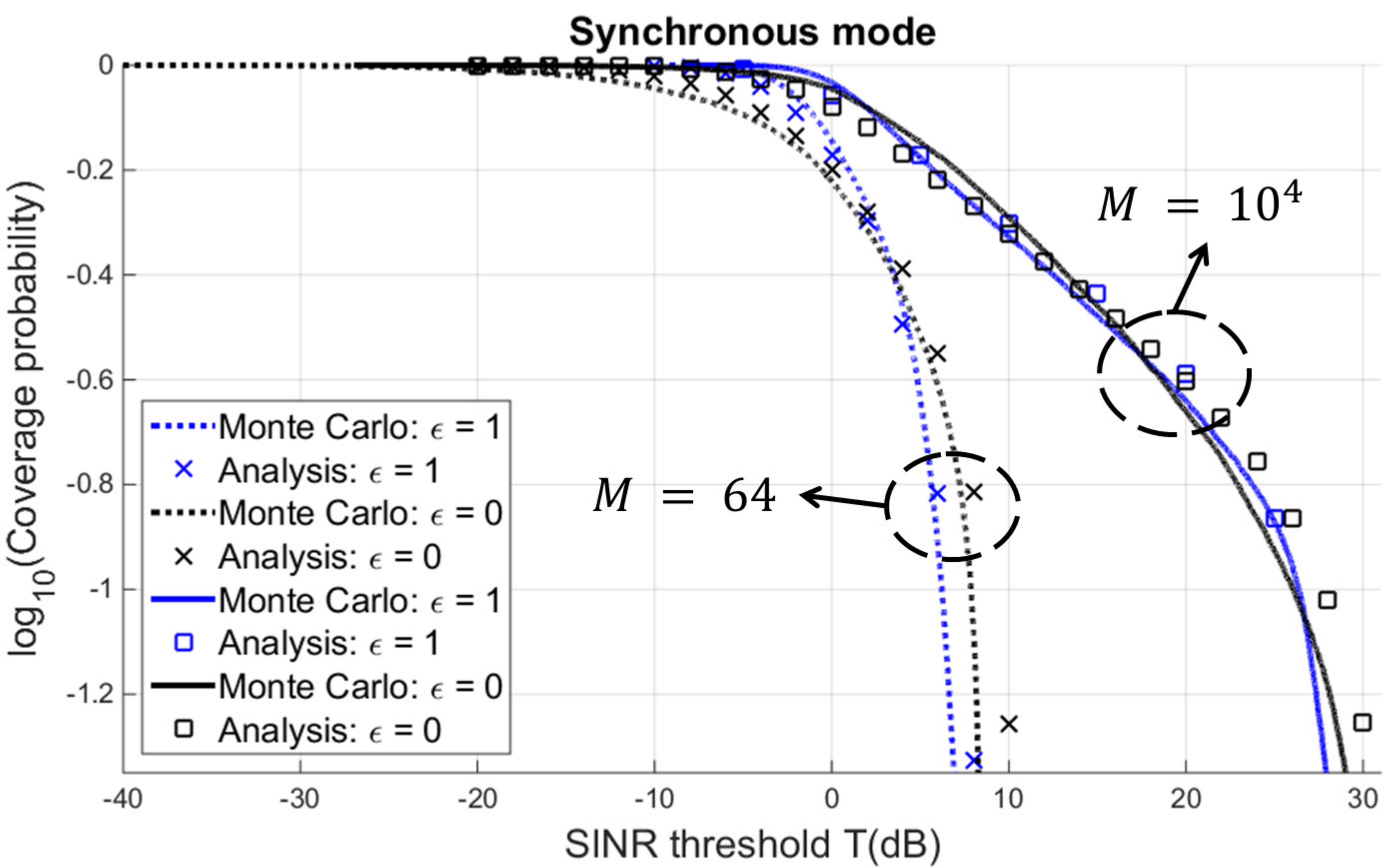}
\caption{Comparison between Monte Carlo and analytical results for different values of ($M$,$\epsilon$) in the synchronous mode.}
\label{syncov}
\end{figure}

Based on the analytical results, we obtain the optimal power control parameter $\epsilon$ and the number of pilot symbols, $N_\text{p}$, in order to maximize the downlink ergodic rate of a cell. We also perform comparisons between the system performance in the asynchronous and synchronous modes.

\textbf{The number of base station antennas:} As seen in Figs. \ref{asyncov} and \ref{syncov}, as the number of base station antennas, $M$, increases, the coverage probability curves in the asynchronous and synchronous modes move to higher values of SINR. Figures \ref{(a)syn_0} and \ref{(a)syn_05} show the coverage probability in both the asynchronous and synchronous modes for $\epsilon = 0$ and $\epsilon = 0.5$, respectively. In these figures, results for different values of $M$ are shown. We observe that as $M$ increases, the gap between the asynchronous and synchronous modes decreases. In Fig. \ref{(a)syn_0}, the coverage probability in the asynchronous and synchronous modes show the same range of SINR for $(M,\epsilon)=(10^4,0)$. Note that for an infinite number of $M$, the dominant source of interference is the pilot contamination. Hence there is small difference between the asynchronous and synchronous modes. In Fig. \ref{(a)syn_05}, it is also observed that for lower SINR thresholds, the synchronous system with $M=64$ outperforms the asynchronous system with $M=10^4$, because as $\epsilon$ increases the pilot contamination dominates the inter-cellular interference for higher values of $M$. 

\begin{figure}[t]
\centering
\includegraphics[width=9cm,height=6cm]{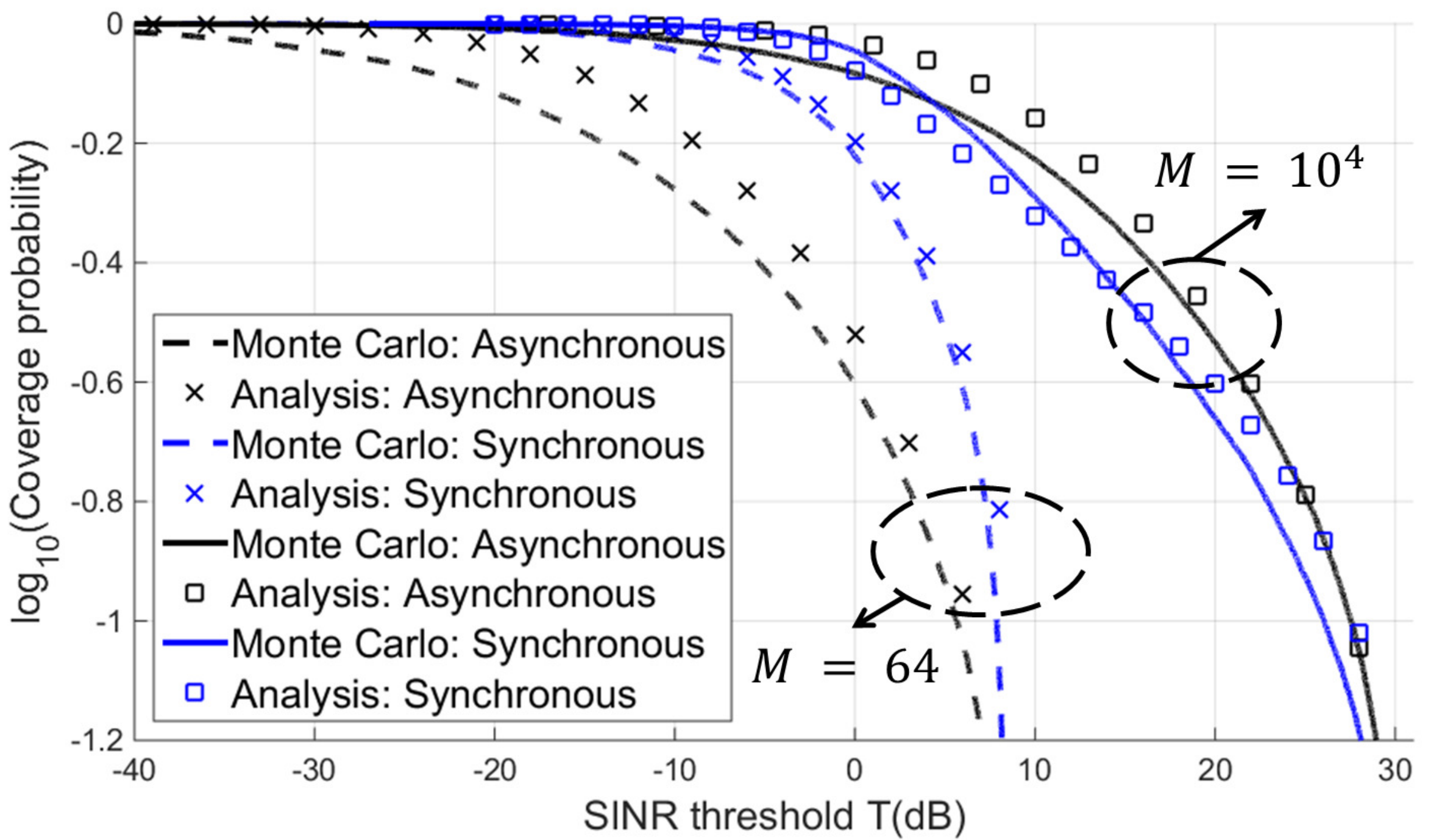}
\caption{Comparison between the asynchronous and synchronous modes for $\epsilon=0$.}
\label{(a)syn_0}
\end{figure}

\begin{figure}[t]
\centering
\includegraphics[width=9cm,height=6cm]{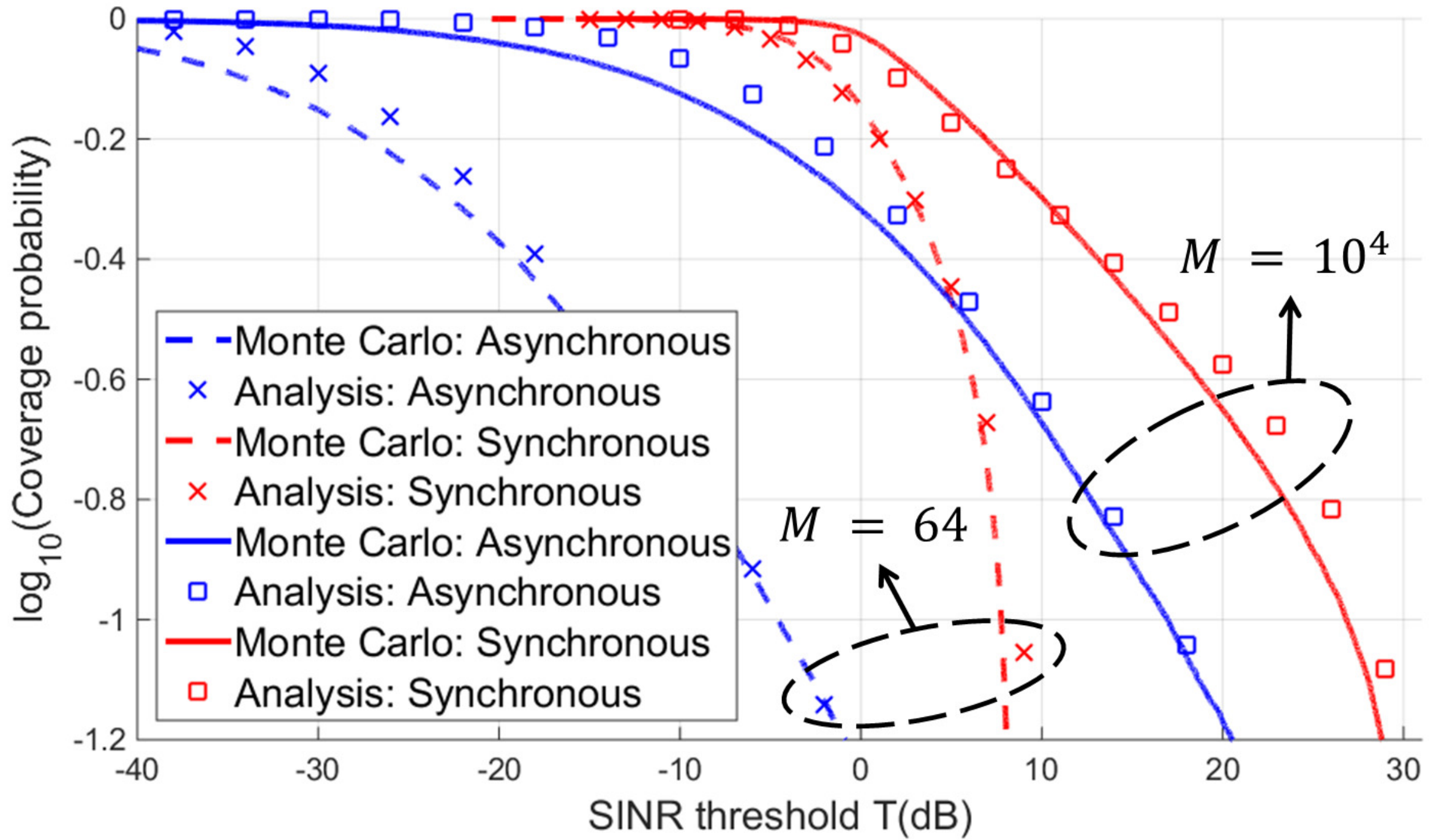}
\caption{Comparison between the asynchronous and synchronous modes for $\epsilon=0.5$.}
\label{(a)syn_05}
\end{figure}

\textbf{The number of pilot symbols:} In Fig. \ref{Np_opt}, the downlink ergodic rate of a cell is shown as a function of $N_\text{p}$ for different values of power control parameter, $\epsilon =[0.2, 0.5, 1]$, in the asynchronous and synchronous modes. In both modes, we observe that as $N_\text{p}$ increases, first, the ergodic rate of a cell increases and then it decreases. This is intuitive because as the number of pilot symbols increases the number of users is also increased but, on the other hand, the interference becomes stronger and also there are less resources for information transmission. Thus, there is a tradeoff and the maximum rate is achieved with a finite number of pilot symbols.

\begin{figure}[t]
\centering
\includegraphics[width=9cm,height=6cm]{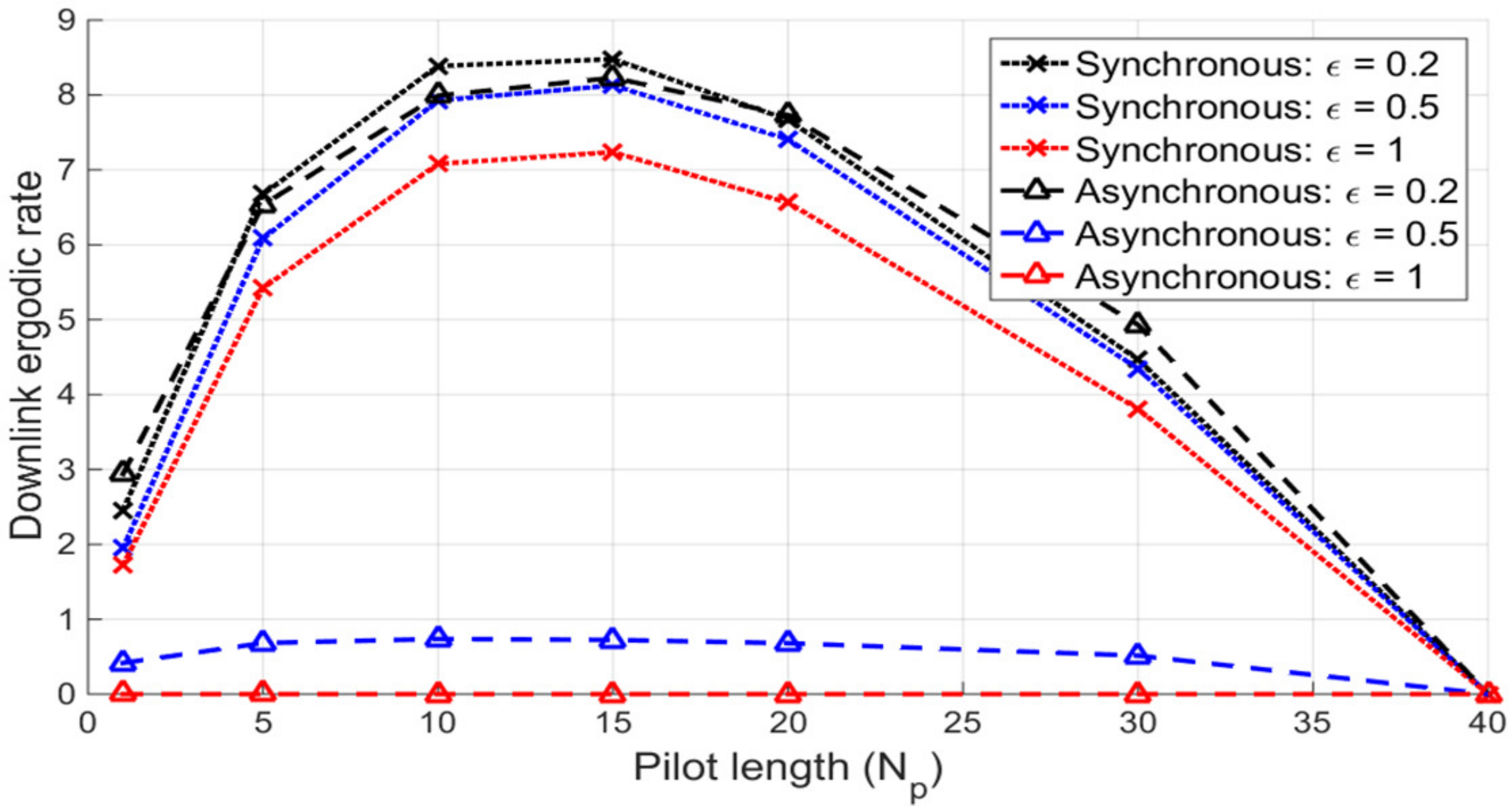}
\caption{Downlink ergodic rate of a cell versus pilot length $N_\text{p}$ for different values of the power control parameter, $\epsilon=[0.2, 0.5, 1]$, in the asynchronous and synchronous modes, $M=64$.}
\label{Np_opt}
\end{figure}

\textbf{Uplink power control parameter:} In Fig. \ref{eps_opt}, the downlink ergodic rate of a cell is shown as a function of the uplink power control parameter for different values of pilot lengths, $N_\text{p} = [5, 10, 30]$, in the asynchronous and synchronous modes. 
As seen, the uplink power control affects the downlink ergodic rate of a cell in two ways. 
First, since we consider that pilots are transmitted in the uplink direction, the uplink power control has impact on the channel estimation performance. Subsequently, the channel estimation performance affects the precoding vector in the downlink transmission. 
Second, in the asynchronous mode, 
the more is the uplink power control parameter, the more is the interference power of the users who are transmitting in the uplink direction on the downlink phase. 
Because of the inter-cellular interference in the uplink, the impact of $\epsilon$ is more perceptible in the asynchronous mode. 
Therefore, in Figs. \ref{Np_opt} and \ref{eps_opt}, we observe that using the full power control, i.e., $\epsilon=1$, leads to no downlink rate in the asynchronous mode. In contrast, in the synchronous mode, we observe that the uplink power control has little effect on the downlink rate.

\begin{figure}[t]
\centering
\includegraphics[width=9cm,height=6cm]{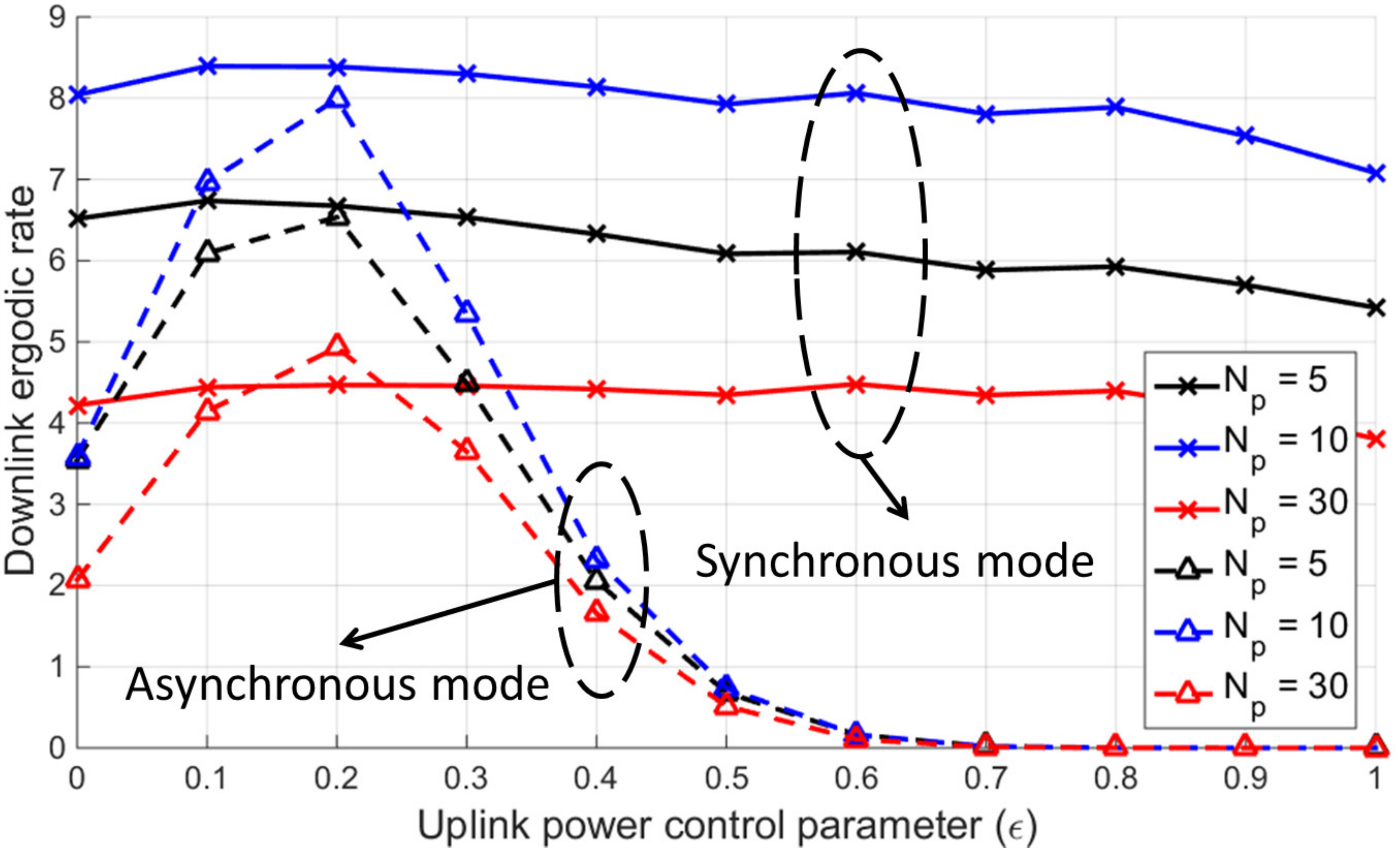}
\caption{Downlink ergodic rate of a cell versus power control parameter $\epsilon$ for different values of pilot lengths, $N_\text{p}=[5, 10, 30]$, in the asynchronous and synchronous modes, $M=64$.}
\label{eps_opt}
\end{figure}

\textbf{Comparison between the asynchronous and synchronous modes:} When we change the system mode from the synchronous mode to the asynchronous mode, we observe two effects. First, some cells have uplink interference rather than downlink interference, which can have higher or less interference than that in the synchronous mode. Second, as also mentioned in \cite{noncoop}, the pilot contamination effect is reduced. In Figs.  \ref{Np_opt} and  \ref{eps_opt}, we observe that for all considered parameter settings, except for $[\epsilon,N_\text{p}]=[0.2,30]$, lower rate is achieved in the asynchronous mode, compared to the synchronous mode. In fact, the reduction of pilot contamination dominates the addition of inter-cellular interference in the uplink direction at the point $[\epsilon,N_\text{p}]=[0.2,30]$. Hence, depending on the parameter settings, higher rate may be achieved by either synchronous or asynchronous mode.  Accordingly, synchronous assumption is not necessarily the worst case for the considered setup.

\section{Conclusion}\label{conclude_sec}
In this paper, downlink massive MIMO in the asynchronous and synchronous modes was analyzed. In the synchrnonous mode, the data transmission in all cells are synchronous. In the asynchronous mode, however, only the transmissions within each cell are synchronous and the transmissions in different cells are asynchronous. In the asynchronous mode, there are more interference sources, i.e., the interference among different base stations and the interference among different users. We used stochastic geometry tool to obtain analytical results for coverage probability and the ergodic rate of the downlink massive MIMO system in both modes. We investigated the system performance for different values of uplink power control parameter and the number of pilot symbols. In this way, we observed that there is an optimal value for the number of pilot symbols. We also saw that the asynchronous system is more sensitive to the uplink power control parameter than is the synchronous mode. We also compared the system performance in the asynchronous and synchronous modes, and observed that the synchronous assumption is not necessarily the worst case for the considered setup, and depending on the parameter settings, higher rates may be achieved by the synchronous or the asynchronous transmission modes.

\appendices
\section{SINR Calculation}\label{sinrcal}
We assume that the desired user, i.e., the $k$-th user of the $l$-th cell, konws the value of $\mathbb{E}\left \{\mathbf{h}_{llk}^T \mathbf{w}_{lk} \right \}$. Hence, the signal power is obtained as follows
\begin{eqnarray}\label{sig}
S&=&\left ( \mathbb{E}\left \{\mathbf{h}_{llk}^T \mathbf{w}_{lk} \right \} \right )^2 =\left( \mathbb{E}\left \{ \mathbf{h}_{llk}^T \frac{\mathbf{u}_{llk}^*}{\norm{\mathbf{u}_{llk}}}\right \}\right )^2
=\left(\mathbb{E}\left \{ \widetilde{\mathbf{h}}_{llk}^T \frac{\mathbf{u}_{llk}^\ast}{\norm{\mathbf{u}_{llk}}}\right \}+\mathbb{E}\left \{ \overline{\mathbf{h}}_{llk}^T \frac{\mathbf{u}_{llk}^\ast}{\norm{\mathbf{u}_{llk}}}\right \}\right)^2\nonumber \\
&\overset{(f)}{=}&
\left(\mathbb{E}\left \{
\frac{\sqrt{P_{lk}}\beta_{llk}}{\mathrm{\Delta}_{lk}}\mathbf{u}_{llk}^T\frac{\mathbf{u}_{llk}^\ast}{\norm{\mathbf{u}_{llk}}}\right \}\right)^2
=\frac{P_{lk} \beta_{llk}^2}{\mathrm{\Delta}_{lk}}\left ( \mathbb{E} \left\{ \theta \right \} \right )^2=\frac{P_{lk} \beta_{llk}^2}{\mathrm{\Delta}_{lk}} C_M^2,
\end{eqnarray}
where $(f)$ follows from the orthogonality of $\widetilde{\mathbf{h}}_{llk}$ and $\mathbf{u}_{llk}$, as well as replacing LMMSE estimation of $\overline{\mathbf{h}}_{llk}$. In addition, in (\ref{sig}), we have $\theta=\sqrt{\sum_{m=1}^{M}\left|u_m\right|^2}$ where $\forall m=1,...,M,\,u_m$ are IID random variables with distribution $\mathcal{CN}\left(0,1\right)$. Finally, according to [18, Sec. \RNum{4}], we have $C_M=\mathbb{E}\left\{\theta\right\}=\frac{\Gamma\left(M+0.5\right)}{\Gamma\left(M\right)}$.

For the first interference term in (\ref{sinr}), we have
\begin{eqnarray}\label{sd_I1}
&&\mathrm{var}\left \{ \mathbf{h}_{llk}^T \mathbf{w}_{lk} \right \}=\mathbb{E}\left |\mathbf{h}_{llk}^T \mathbf{w}_{lk} \right |^2 -\left (\mathbb{E} \left \{\mathbf{h}_{llk}^T \mathbf{w}_{lk} \right \}\right)^2\nonumber\\
&&\overset{(g)}{=}\mathbb{E}\norm{\overline{\mathbf{h}}_{llk} }^2+\mathbb{E}\left |\widetilde{\mathbf{h}}_{llk}^T \frac{\mathbf{u}_{llk}^\ast}{\norm{\mathbf{u}_{llk}}} \right |^2-\frac{P_{lk} \beta_{llk}^2}{\mathrm{\Delta}_{lk}}\left ( \mathbb{E} \left \{ \theta \right \} \right )^2\nonumber\\
&&=\frac{P_{lk} \beta_{llk}^2}{\mathrm{\Delta}_{lk}} \mathbb{E}\{\theta^2\}+\left(\beta_{llk}-\frac{P_{lk}\beta_{llk}^2}{\mathrm{\Delta}_{lk}}\right )-\frac{P_{lk} \beta_{llk}^2}{\mathrm{\Delta}_{lk}}\left (\mathbb{E} \left \{ \theta \right\} \right )^2
=\frac{P_{lk} \beta_{llk}^2}{\mathrm{\Delta}_{lk}}\left(V_M-1\right)+\beta_{llk},
\end{eqnarray}
where $(g)$ is obtained from rewriting $\mathbf{h}_{llk}$ as $\overline{\mathbf{h}}_{llk}+\widetilde{\mathbf{h}}_{llk}$. Additionally, $V_M$ denotes the varience of $\theta$. It is straightforward to show that $\mathbb{E}\{\theta^2\}=M$. Hence, we have $V_M=M-C_M^2$.

Since the channel vector of the desired cell, $\mathbf{h}_{llk}$, and the observation vectors of the co-cell users, $\forall k'\ne k\,\mathbf{u}_{llk'}$ are uncorrelated, the second term of the interference in (\ref{sinr}) is
\begin{eqnarray}\label{sd_I2}
\sum_{k'\ne k} \mathbb{E}\left|\mathbf{h}_{llk}^T\mathbf{w}_{lk'}\right|^2=\left(N_\text{p}-1\right)\beta_{llk}.
\end{eqnarray}

Then, for the third and the forth terms of the interference, we have to obtain $\mathbb{E}\left|\mathbf{h}_{jlk}^T\mathbf{w}_{jk'}\right|^2$, which is given by
\begin{eqnarray}\label{two_comp}
\mathbb{E}\left|\mathbf{h}_{jlk}^T\mathbf{w}_{jk'}\right|^2=
\mathbb{E}\left|\widetilde{\mathbf{h}}_{jlk}^T\frac{\mathbf{u}_{jjk'}^*}{\norm{\mathbf{u}_{jjk'}}}\right|^2+
\mathbb{E}\biggl|\overline{\mathbf{h}}_{jlk}^T\frac{\mathbf{u}_{jjk'}^*}{\norm{\mathbf{u}_{jjk'}}}\biggr|^2.
\end{eqnarray}

Due to the orthogonality of $\widetilde{\mathbf{h}}_{jlk}$ and $\mathbf{w}_{jk'}$, as well as the distribution of the channel estimation error in (\ref{ch_err1}), the first term of (\ref{two_comp}) is found as
\begin{eqnarray}
\mathbb{E}\left|\widetilde{\mathbf{h}}_{jlk}^T\mathbf{w}_{jk'}\right|^2
=\beta_{jlk}-
P_{lk} \beta_{jlk}^2
\begin{cases}
\frac{1}{\mathrm{\Delta}_{jk}}
&j\in S_l,\\
\sum_{k''=1}^K\frac{N_\text{p}+N_\text{u}}{N_{\text{tot}}^2}\frac{1}{\mathrm{\Delta}_{jk''}}
& j\not\in S_l.\end{cases}
\end{eqnarray}

Additionally, by replacing $\overline{\mathbf{h}}_{jlk}$ from (\ref{ch_est}), the second term of (\ref{two_comp}) is found as
\begin{eqnarray}
\mathbb{E}\left|\overline{\mathbf{h}}_{jlk}^T\mathbf{w}_{jk'}\right|^2
=P_{lk} \beta_{jlk}^2
\begin{cases}
\frac{M-1}{\mathrm{\Delta}_{jk}}\delta\left(k,k'\right)+\frac{1}{\mathrm{\Delta}_{jk}}
&j\in S_l,\\
\frac{N_\text{p}+N_\text{u}}{N_{\text{tot}}^2}\left(\frac{M-1}{\mathrm{\Delta}_{jk'}}+\sum\limits_{k''=1}^K\frac{1}{\mathrm{\Delta}_{jk''}}\right)
& j\not\in S_l.\end{cases}
\end{eqnarray}

Thus, we obtain
\begin{eqnarray}\label{sd_I3}
\mathbb{E}\left|\mathbf{h}_{jlk}^T\mathbf{w}_{jk'}\right|^2
=\beta_{jlk}+\left(M-1\right)\frac{P_{lk} \beta_{jlk}^2}{\mathrm{\Delta}_{jk'}}
\begin{cases}
\delta\left(k,k'\right)
&j\in S_l,\\
\frac{N_\text{p}+N_\text{u}}{N_{\text{tot}}^2}
& j\not\in S_l.\end{cases}
\end{eqnarray}

Finally, by doing some calculations and using the results of (\ref{sig}), (\ref{sd_I1}), (\ref{sd_I2}), and (\ref{sd_I3}), $\mathrm{SINR}^{-1}$ is achieved as in (\ref{sinr_1}).

\section{Proof of Theorem \ref{thm1}}\label{pthm1}
First, as we consider a cell with the exclusion of a central disk of radius $r_0$ around its base station, we need to obtain $\mathbb{P}\left(r>R\right|r>r_0)$. According to the Bayes' rule, $\mathbb{P}\left(r>R |r>r_0\right)$ can be expressed as
\begin{eqnarray}\mathbb{P}\left(r>R|r>r_0\right)=\frac{\mathbb{P}\left(r>R,\,r>r_0\right)}{\mathbb{P}\left(r>r_0\right)}.\end{eqnarray}

Since the desired user's serving base station is the nearest base station to the desired user, $\mathbb{P}\left(r>r_0\right)$ is obtained by the fact that there is no base station in a distance less than $r_0$ to the desired user. Thus, 
$\mathbb{P}\left(r>r_0\right)=e^{-\pi\lambda r_0^2}$. By similar arguments, $\mathbb{P}\left(r>R,\,r>r_0\right)$ can be express as $e^{-\pi\lambda \left(\max\{R,r_0\}\right)^2}$. Therefore, we have
\begin{eqnarray}\label{thm1_bef_last}
\mathbb{P}\left(r>R|r>r_0\right)=\begin{cases}\frac{e^{-\pi \lambda R^2}}{e^{-\pi \lambda r_0^2}}&r_0<R,\\1&r_0>R.\end{cases}
\end{eqnarray}
Finally, by taking derivative from (\ref{thm1_bef_last}), we obtain $f\left(r\right)$ as in (\ref{thm1_eq}).
%%%%%%%%%%%%%%%%%%%%%%%%%%%%%%%%%%%%%%%%%%

%Proof of Theorem2%%%%%%%%%%%%%%%%%%%%%%%%%%%%%%%%
\section{Proof of Theorem \ref{thm2}}\label{pthm2}
At first, we obtain $\mathbb{P}\left(r_1>R\right|r_1>r_0,\, r_2)$. As the serving base station of a user is the nearest one to that user, we have $r_1< r_2$ and, consequently,
\begin{eqnarray}\label{sameproof}
\mathbb{P}\left(r_1>R\right|r_1>r_0,\, r_2)=\mathbb{P}\left(r_1>R|r_0<r_1<r_2\right).
\end{eqnarray}

Based on the Bayes' rule, we have
\begin{eqnarray}
\mathbb{P}\left(r_1>R\right|r_0<r_1<r_2)=\frac{\mathbb{P}\left(r_0<r_1<r_2,\,R<r_1\right)}{\mathbb{P}\left(r_0<r_1<r_2\right)}
=\begin{cases}
1&r_0>R,\\
\frac{\mathbb{P}\left(R<r_1<r_2\right)}{\mathbb{P}\left(r_0<r_1<r_2\right)}&r_0<R<r_2,\\
0&r_2<R.
\end{cases}
\end{eqnarray}

The propability $\mathbb{P}\left(r_0<r_1<r_2\right)$ can be obtained by the fact that there is no base station within a radius of $r_0$ around the desired user, and also there will be at least one base station within the area between the two circles of radii $r_1$ and $r_2$ around the desired user (Fig. \ref{thm2fig}). Hence, we have
\begin{eqnarray}
\label{53}
\mathbb{P}\left(r_0<r_1<r_2\right)=e^{-\pi\lambda r_0^2}\left(1-e^{-\pi\lambda\left(r_2^2-r_0^2\right)}\right)
=e^{-\pi\lambda r_0^2}-e^{-\pi\lambda r_2^2}.
\end{eqnarray}

\begin{figure}[!t]
\centering
\includegraphics[width=4cm,height=4.5cm]{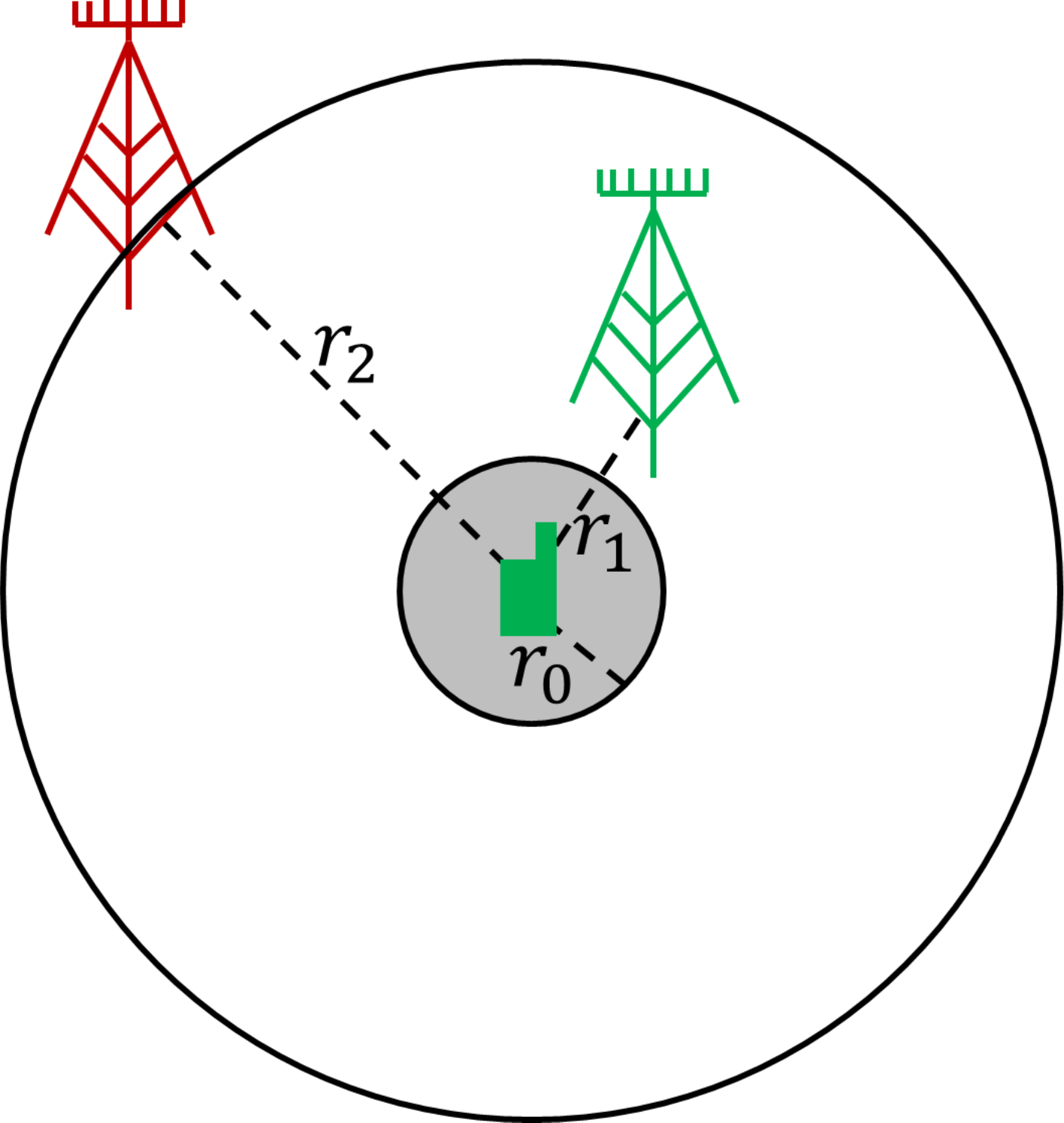}
\caption{The desired user's situation between its serving base station and another arbitrary base station.}
\label{thm2fig}
\end{figure}

Also, $\mathbb{P}\left(r_0<r_1<r_2,\,R<r_1\right)$ can be obtained  with the same procedure as in (\ref{53}). 
Then, we have
\begin{eqnarray}\label{similar}
\mathbb{P}\left(r_1>R|r_0<r_1<r_2\right)=\begin{cases}1&r_0>R,\\ \frac{e^{-\pi\lambda R^2}-e^{-\pi\lambda r_2^2}}{e^{-\pi\lambda r_0^2}-e^{-\pi\lambda r_2^2}}&r_0<R<r_2,\\0&r_2<R.\end{cases}
\end{eqnarray}

Finally, by doing some calculations, the final result is given by
\begin{eqnarray}
f(r_1|r_2)=\mathbb{P}\left(r_1|r_1>r_0,\,r_2\right)
&=&\frac{\mathrm{d}\left(1-\mathbb{P}\left(r_1>R|r_1>r_0\,r_2\right)\right)}{\mathrm{d}R}\nonumber\\
&=&\frac{2\pi\lambda r_1 e^{-\pi \lambda r_1^2}}{e^{-\pi\lambda r_0^2}-e^{-\pi\lambda r_2^2}},\, r_0<r_1<r_2.
\end{eqnarray}
%%%%%%%%%%%%%%%%%%%%%%%%%%%%%%%%%%%%%%%%%%

%Proof of Theorem3%%%%%%%%%%%%%%%%%%%%%%%%%%%%%%%%
\section{Proof of Theorem \ref{thm3}}\label{pthm3}
Figure \ref{thmfig1} indicates the situation of two users of different cells. According to Fig. \ref{thmfig1} and the triangle inequality, we have
\begin{eqnarray}\label{ineq1}
r_{ljk'}<r_{lkjk'}+r_{llk},\quad r_{jlk}<r_{lkjk'}+r_{jjk'}.
\end{eqnarray}

\begin{figure}[!t]
\centering
\includegraphics[width=5cm,height=4.5cm]{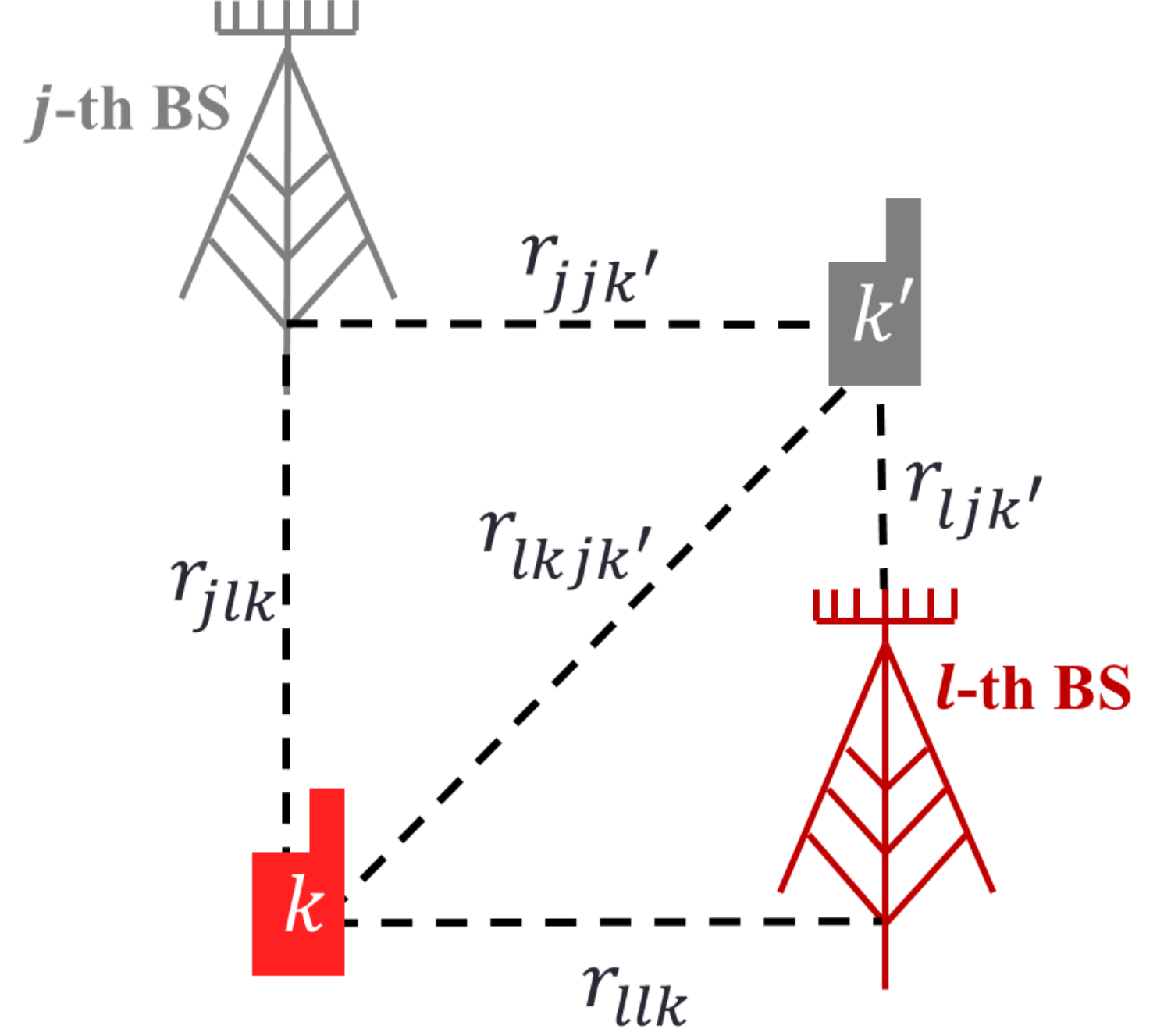}
\caption{The situation of two users of different cells and their serving base stations.}
\label{thmfig1}
\end{figure}

Also, since each user is served by the nearest base station, the following inequalities hold
\begin{eqnarray}\label{ineq2}
r_{jjk'}<r_{ljk'},\quad r_{llk}<r_{jlk}.
\end{eqnarray}

From the inequalities (\ref{ineq1}) and (\ref{ineq2}) as well as the fact that the distance between a user and its serving base station is greater than $r_0$, it is concluded that
\begin{eqnarray}\max{\left(r_0,\,r_{llk}-r_{lkjk'}\right)}<r_{jjk'}<r_{llk}+r_{lkjk'}.\end{eqnarray}

Therefore, we have
\begin{eqnarray}\label{thm3equ1}
\mathbb{P}\left(r_{jjk'}=s|r_{lkjk'}=r,\,r_{llk}=x\right)
&=&\mathbb{P}\left(r_{jjk'}=s|\max{\left(r_0,\,x-r\right)}<r_{jjk'}<x+r\right)\nonumber\\
&=&\mathbb{P}\left(r_{jjk'}=s|R_1<r_{jjk'}<R_2\right).\end{eqnarray}

The expression $\mathbb{P}\left(r_{jjk'}>s|R_1<r_{jjk'}<R_2\right)$ and (\ref{sameproof}) are similar. Therefore, by following the same procedure as in (\ref{similar}), we have
\begin{eqnarray}
\mathbb{P}\left(r_{jjk'}>s|R_1<r_{jjk'}<R_2\right)=
\begin{cases}
1&R_1>s,\\
 \dfrac{e^{-\pi\lambda s^2}-e^{-\pi\lambda R_2^2}}{e^{-\pi\lambda R_1^2}-e^{-\pi\lambda R_2^2}}&R_1<s<R_2,\\
 0&R_2<s.
 \end{cases}
\end{eqnarray}

Finally, we obtain
\begin{eqnarray}
\mathbb{P}\left(r_{jjk'}=s|r_{lkjk'}=r,\,r_{llk}=x\right)
&=&\frac{\mathrm{d}}{\mathrm{d}s}\left(1-\mathbb{P}\left(r_{jjk'}>s|r_{lkjk'}=r,\,r_{llk}=x\right)\right)\nonumber\\
&=&\frac{2\pi\lambda s e^{-\pi\lambda s^2}}{e^{-\pi\lambda R_1^2}-e^{-\pi\lambda R_2^2}},\quad R_1<s<R_2.
\end{eqnarray}
%%%%%%%%%%%%%%%%%%%%%%%%%%%%%%%%%%%%%%%%%%

%Approximation%%%%%%%%%%%%%%%%%%%%%%%%%%%%%%%%%%%
\section{Gamma Random Variable CDF Approximation}\label{gam_appr}
Consider that $A$ is a random variable. 
We can present $\mathbb{P}\left(g>A\right)$ by taking expectation as follows.
\begin{eqnarray}
\label{exp_form}
\mathbb{P}\left(g>A\right)&
=&1-\int \mathbb{P}\left(g<A|A\right)\mathbb{P}\left(A\right)\, \mathrm{d}A
=1-\mathbb{E}_A\left\{F_g\left(A\right)\right\}\nonumber\\
&\overset{(h)}{\approx}&\sum_{n=1}^N \left(-1\right)^{n+1} \binom{N}{n} \mathbb{E}_{A}\left\{e^{-\eta nA}\right\},
\end{eqnarray}
where $F_g\left(A\right)$ is the cumulative distribution function (CDF) of the Gamma random variable, and $\eta=N\left(N!\right)^{-\frac{1}{N}}$. In $(h)$, based on the alzer inequality $ F_g\left(A\right)\leq\left(1-e^{-\eta A}\right)^N$ in [20, Appendix A] and \cite{alzer2}, a tight CDF approximation for the Gamma random variable, $ F_g\left(A\right)\approx\left(1-e^{-\eta A}\right)^N$, is used.
%Coverage probability%%%%%%%%%%%%%%%%%%%%%%%%%%%%%%%%
\section{Coverage Probability Calculation}\label{cov_cal}
First, we obtain $Q_1$. In the asynchronous mode, $S_{l}$ only includes ${l}$. Thus, we have
\begin{eqnarray}
Q_1^\text{asyn}
=\frac{N_\text{p}+N_\text{u}}{N_\text{tot}^2}N_\text{p}\mathbb{E}_{\{r_{jjk''},r_{ljk''}|j\in\mathrm{\Phi}_\text{b}\backslash\{l\}\}\big|r_{llk}=x}\left\{\sum_{j\ne l}r_{jjk''}^{\alpha\epsilon}r_{ljk''}^{-\alpha}\right\}
+\frac{P_\text{d}N_\text{p}N_\text{d}Q_2^\text{asyn}}{P_\text{u}\omega^{-\epsilon}N_\text{tot}^2}+\frac{\sigma^2\omega^{\epsilon-1}}{N_\text{p}P_\text{u}},
\end{eqnarray}
where $\mathbb{E}\left\{\sum_{j\not\in S_l}r_{lj}^{-\alpha}\right\}$ is replaced by 
$Q_2^\text{asyn}$. 
Similarly, in the synchronous mode, we have
\begin{eqnarray}
Q_1^\text{syn}=\mathbb{E}_{\{r_{jjk},r_{ljk}|j\in\mathrm{\Phi}_\text{b}\backslash\{l\}\}\big|r_{llk}=x}\left\{\sum_{j\ne l} r_{jjk}^{\alpha \epsilon} r_{ljk}^{-\alpha}\right\}
+\frac{\sigma^2}{N_\text{p}P_\text{u}\omega^{1-\epsilon}}.
\end{eqnarray}

In both modes, we should obtain $\mathbb{E}\left\{\sum_{j\ne l} r_{jjk}^{\alpha \epsilon} r_{ljk}^{-\alpha}|r_{llk}=x\right\}$ to find out $Q_1$. Then,
\begin{eqnarray}
\label{67}
&&\mathbb{E}_{\{r_{jjk},r_{ljk}|j\in\mathrm{\Phi}_\text{b}\backslash\{l\}\}\big|r_{llk}=x}\left\{\sum_{j\ne l} r_{jjk}^{\alpha \epsilon} r_{ljk}^{-\alpha} \right\}
=\mathbb{E}_{\mathrm{\Phi}_{lk}^\text{u}}\left\{\sum_{j\ne l} \mathbb{E}_{r_{jjk}|r_{ljk}}\left\{r_{jjk}^{\alpha \epsilon}\right\}r_{ljk}^{-\alpha}\right\}\nonumber\\
&&\overset{(i)}{=}\mathbb{E}_{\mathrm{\Phi}_{lk}^\text{u}}\left\{\sum_{j\ne l} r_{ljk}^{-\alpha}\int_{r_0}^{r_{ljk}}\frac{ r_{jjk}^{\alpha \epsilon}2\pi \lambda r_{jjk}e^{-\pi \lambda r_{jjk}^2}}{e^{-\pi \lambda r_0^2}-e^{-\pi \lambda r_{ljk}^2}}\,\mathrm{d}r_{jjk}\right\}\nonumber\\
&&\overset{(j)}{=}\mathbb{E}_{\mathrm{\Phi}_{lk}^\text{u}}\left\{\sum_{j\ne l}\int_{\pi \lambda r_0^2}^{\pi \lambda r_{ljk}^2}\frac{R_\text{e}^{\alpha \epsilon} r_{ljk}^{-\alpha}s^{\frac{\alpha\epsilon}{2}}e^{-s}}{e^{-\pi \lambda r_0^2}-e^{-\pi \lambda r_{ljk}^2}}\,\mathrm{d}s\right\}\nonumber\\
&&\overset{(k)}{=}R_\text{e}^{\alpha \epsilon}\int_{R_\text{e}}^\infty \int_{\pi \lambda r_0^2}^{\pi \lambda r^2}
\frac{r^{-\alpha}s^{\frac{\alpha\epsilon}{2}}e^{-s}2\pi \lambda r}{e^{-\pi \lambda r_0^2}-e^{-\pi \lambda r^2}}\,\mathrm{d}s\,\,\mathrm{d}r
\overset{(l)}{=}R_\text{e}^{-\alpha\left(1-\epsilon\right)}\int_{\pi \lambda R_\text{e}^2}^\infty \int_{\pi \lambda r_0^2}^t 
\frac{t^\frac{-\alpha}{2}s^{\frac{\alpha\epsilon}{2}}e^{-s}}{e^{-\pi \lambda r_0^2}-e^{-t}}\,\mathrm{d}s\,\mathrm{d}t,
\end{eqnarray}
where in $(i)$, the distance distribution given in Theorem \ref{thm2} is used. In $(j)$, we use $R_\text{e}=\left(\pi\lambda\right)^{-\frac{1}{2}}$ and the variable transform $s=\pi\lambda r_{jjk}^2$. $(k)$ comes from the Campbell's theory and the exclusion ball model. In $(l)$, the variable changes as $t=\pi\lambda r^2$.

Besides the expression in (\ref{67}), obtaining $Q_1^\text{asyn}$ involves the calculation of $Q_2^\text{asyn}$. From (\ref{Q_2}), in the synchronous mode, $Q_2^\text{syn}$ is zero. In the asynchronous mode, we get
\begin{eqnarray}
Q_2^\text{asyn}&=&\mathbb{E}_{\{r_{lj}|j\in\mathrm{\Phi}_\text{b}\backslash \{l\}\}\big|r_{llk}=x}\left\{\sum\limits_{j\ne l}r_{lj}^{-\alpha}\right\}\approx
\mathbb{E}_{\{r_{lj}|j\in\mathrm{\Phi}_\text{b}\backslash \{l\}\}}\left\{\sum\limits_{j\ne l}r_{lj}^{-\alpha}\right\}\nonumber\\
&\overset{(m)}{\approx}&\int_{R_\text{e}}^\infty r^{1-\alpha}2\pi\lambda \,\mathrm{d}r
\overset{(n)}{=}\frac{2R_\text{e}^{-\alpha}}{\alpha-2},
\end{eqnarray}
where in $(m)$, for simplicity, we consider that the interfering base stations are distributed based on the exclusion ball model. Thus, by using the Campbell's theory, $(m)$ is found. In addition, in $(n)$, we assume $R_\text{e}=\left(\pi\lambda\right)^{-\frac{1}{2}}$.

Next, it is clear that $Q_3^\text{sync}=0$ in the synchronous mode. For the asynchronous mode, it is obtained as
\begin{eqnarray}\label{difficult}
Q_3^\text{asyn}&=&\mathbb{E}_{\{r_{jjk'},r_{lkjk'}|j\in\mathrm{\Phi}_\text{b}\backslash\{l\},k'=1,2,...,K\}\big|r_{llk}=x}\left\{\sum_{j\ne l}\sum_{k'} r_{jjk'}^{\alpha\epsilon}r_{lkjk'}^{-\alpha}\right\}\nonumber\\
&=&N_\text{p}\mathbb{E}_{\{r_{jjk'},r_{lkjk'}|j\in\mathrm{\Phi}_\text{b}\backslash\{l\}\}\big|r_{llk}=x}\left\{\sum_{j\ne l}r_{jjk'}^{\alpha\epsilon}r_{lkjk'}^{-\alpha}\right\}\nonumber\\
&=&N_\text{p}\mathbb{E}_{\mathrm{\Phi}_{jk'}^\text{u}|r_{llk}=x}\left\{\sum_{j\ne l}\mathbb{E}_{r_{jjk'}|r_{lkjk'},r_{llk}=x}\left\{r_{jjk'}^{\alpha\epsilon}\right\}r_{lkjk'}^{-\alpha}\right\}\nonumber\\
&\overset{(p)}{=}&N_\text{p}\int_{R_\text{e}}^\infty \int_0^{2\pi}\mathbb{E}_{r_{jjk'}|r_{lkjk'}=r_1,r_{llk}=x}\left\{r_{jjk'}^{\alpha\epsilon}\right\}r_1^{-\alpha}\lambda r\,\mathrm{d}\theta\mathrm{d}r\nonumber\\
&\overset{(q)}{=}&N_\text{p}\int_{R_\text{e}}^\infty \int_0^{2\pi}
A\left(x\right)\int_{-\pi\lambda \left(\max{\left(r_0,\,x-r_1\right)}\right)^2}^{-\pi\lambda \left(x+r_1\right)^2}s^{\frac{\alpha\epsilon}{2}}e^{-s}\,\mathrm{d}s\,
r_1^{-\alpha}\lambda r\,\mathrm{d}\theta\,\mathrm{d}r,
\end{eqnarray}
where $A\left(x\right)=\dfrac{\left(\pi\lambda\right)^{\frac{\alpha\epsilon}{2}}}{e^{-\pi\lambda \left(
\max{\left(r_0,\,x-r_1\right)}
\right)^2}-e^{-\pi\lambda \left(x+r_1\right)^2}}$. 
The expression in $(p)$ is obtained from the Campbell's theory and the exclusion ball model, and in $(q)$, we use the distance distribution given in Theorem \ref{thm3}. Note that, $r_1$ can be expressed based on $r$, $x$, and $\theta$ as $r_1^2=r^2+x^2-2rx\cos\theta$ by using the law of cosines (Fig. \ref{udistance}).
\begin{figure}[t]
\centering
\includegraphics[width=6cm,height=5cm]{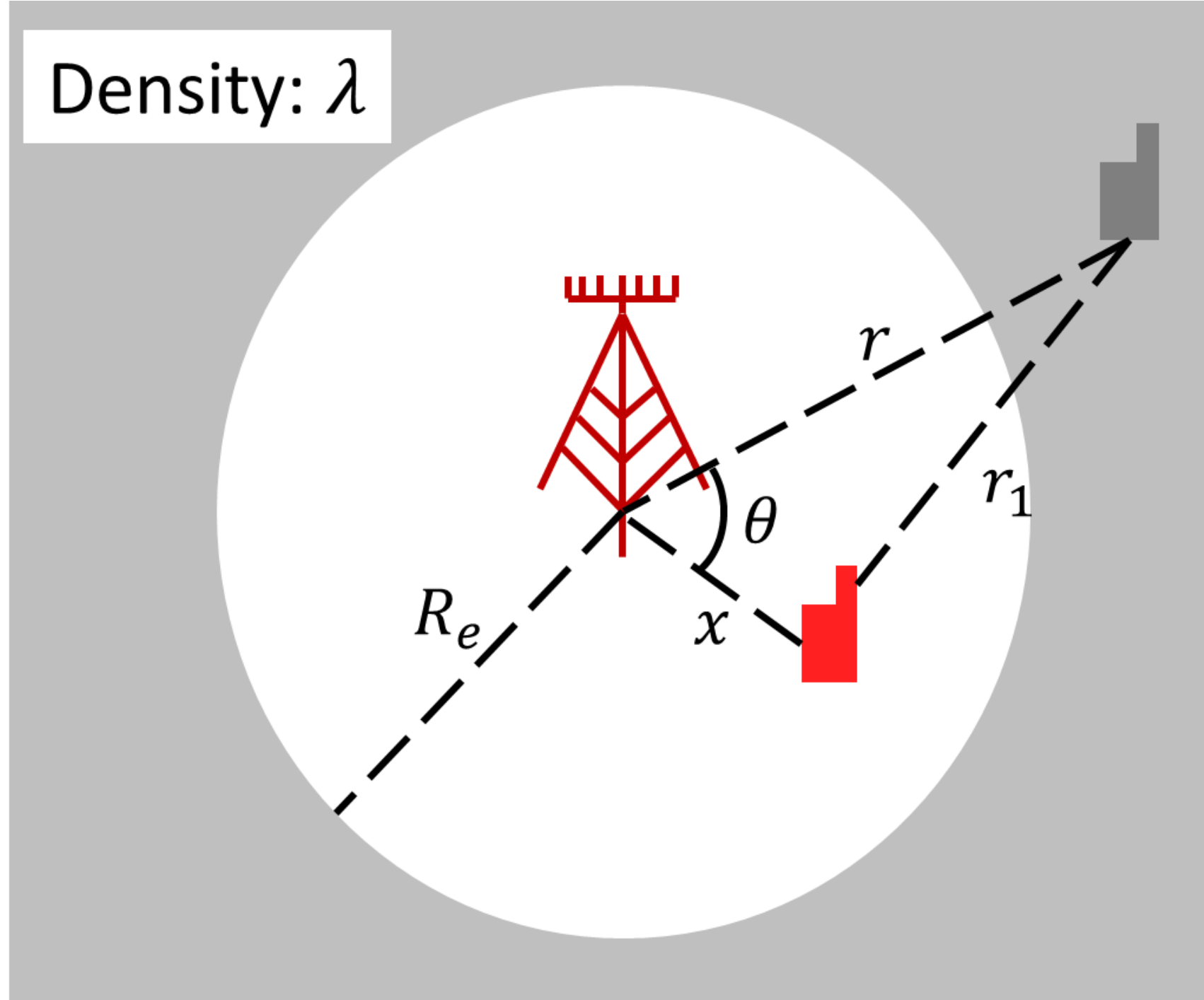}
\caption{The relation between $r_1$ and $r$ can be obtained from the law of cosines.}
\label{udistance}
\end{figure}
In (\ref{difficult}), we deal with a triple integral which does not have close form expression and should be solved numerically. Therefore, such an expression is very difficult to be computed. For simplicity, we use the approximation, $r_{lkjk'} \approx r_{ljk'}$, and from (\ref{67}), we obtain
\begin{eqnarray}\label{easy}
Q_3^\text{asyn}
=N_\text{p}\left(\pi \lambda\right)^{\frac{\alpha}{2}\left(1-\epsilon\right)}
\int_{\pi \lambda R_\text{e}^2}^\infty t^\frac{-\alpha}{2}\int_{\pi \lambda r_0^2}^t s^{\frac{\alpha\epsilon}{2}}\frac{e^{-s}}{e^{-\pi \lambda r_0^2}-e^{-t}}\,\mathrm{d}s\,\mathrm{d}t.
\end{eqnarray}

The next step to find the coverage probability is to obtain $\mathcal{E}_1\left(T,n,x\right)$ and $\mathcal{E}_2\left(T,n,x\right)$. In both modes, $\mathcal{E}_1\left(T,n,x\right)$ can be expressed as bellow.
\begin{eqnarray}\label{72}
&&\mathcal{E}_1\left(T,n,x\right)=\mathbb{E}_{\mathrm{\Phi}_\text{b}\backslash\{l\}|r_{llk}=x}\left\{\exp{\left(\sum_{j\ne l}\left(B\left(T,n,x\right)r_{jlk}^{-\alpha}+C\left(T,n,x\right)r_{jlk}^{-2\alpha}\right)\right)}\right\}\nonumber\\
&&\overset{(r)}{=}\exp\left(\int_{\pi\lambda x^2}^\infty \left[\exp\left(
B\left(T,n,x\right)R_\text{e}^{-\alpha}t^\frac{-\alpha}{2}+C\left(T,n,x\right)R_\text{e}^{-2\alpha}t^{-\alpha}
\right)-1\right]
\,\mathrm{d}t\right),
\end{eqnarray}
where $B\left(T,n,x\right)$ and $C\left(T,n,x\right)$ are defined in the following equations
\begin{eqnarray}
\label{b_asyn}
B^\text{asyn}\left(T,n,x\right)&=&-\eta nT\frac{N_\text{p}}{C_M^2}\frac{N_\text{d}^2}{N_\text{tot}^2}\left(x^\alpha+x^{\alpha \left(2-\epsilon\right)}Q_1^\text{asyn}\right),\\
\label{b_syn}
B^\text{syn}\left(T,n,x\right)&=&-\eta nT\frac{N_\text{p}}{C_M^2}\left(x^\alpha+x^{\alpha \left(2-\epsilon\right)}Q_1^\text{syn}\right),\\
\label{c_asyn}
C^\text{asyn}\left(T,n,x\right)&=&-\eta nT\frac{M-1}{C_M^2}\frac{N_\text{p}N_\text{d}^2\left(N_\text{p}+N_\text{u}\right)}{N_\text{tot}^4}x^{2\alpha},\\
\label{c_syn}
C^\text{syn}\left(T,n,x\right)&=&-\eta nT\frac{M-1}{C_M^2}x^{2\alpha}.
\end{eqnarray}
In the expression (\ref{72}), $(r)$ is achieved from Campbell's theory as well as the exclusion ball model and using the variable transform $t=\pi\lambda r^2$.

$\mathcal{E}_2^\text{asyn}\left(T,n,x\right)$ in the asynchronous mode is obtained as
\begin{eqnarray}
\label{E2_async}
&\mathcal{E}_2^\text{asyn}&\left(T,n,x\right)
=\mathbb{E}_{\{r_{jjk'},r_{ljk'}|j\in\mathrm{\Phi}_\text{b}\backslash\{l\},k'=1,...,K\}\big|r_{llk}=x}\left\{\exp{\left(D^\text{asyn}\left(T,n,x\right)\sum_{j\ne l}\sum_{k'}r_{jjk'}^{\alpha\epsilon}r_{ljk'}^{-\alpha}
\right)}\right\}\nonumber\\
&&=\left(\mathbb{E}_{\{r_{jjk'},r_{ljk'}|j\in\mathrm{\Phi}_\text{b}\backslash\{l\}\}\big|r_{llk}=x}\left\{\prod_{j\ne l}\exp{\left(D^\text{asyn}\left(T,n,x\right)
r_{jjk'}^{\alpha\epsilon}r_{ljk'}^{-\alpha}
\right)}\right\}\right)^{N_\text{p}}\nonumber\\
&&=\left(\mathbb{E}_{\mathrm{\Phi}_{jk'}^\text{u}|r_{llk}=x}\left\{\prod_{j\ne l}\mathbb{E}_{r_{jjk'}|r_{ljk'},r_{llk}=x}\left\{\exp{\bigl(D^\text{asyn}\left(T,n,x\right)
r_{jjk'}^{\alpha\epsilon}r_{ljk'}^{-\alpha}
\bigr)}\right\}
\right\}\right)^{N_\text{p}}\nonumber\\
&&\overset{(s)}{=}\biggl(\mathbb{E}_{\mathrm{\Phi}_{jk'}^\text{u}|r_{llk}=x}\biggl\{\prod_{j\ne l}
\int_{r_0}^{r_{ljk'}} \frac{2\pi\lambda r e^{-\pi\lambda r^2}}{e^{-\pi\lambda r_0^2}-e^{-\pi\lambda r_{ljk'}^{-\alpha} }}\exp{\bigl(D^\text{asyn}\left(T,n,x\right)
r^{\alpha\epsilon}r_{ljk'}^{-\alpha}
\bigr)}\,\mathrm{d}r
\biggr\}\biggr)^{N_\text{p}}\nonumber\\
&&\overset{(t)}{=}\biggl(\mathbb{E}_{\mathrm{\Phi}_{jk'}^\text{u}|r_{llk}=x}\biggl\{\prod_{j\ne l}
\int_{\pi\lambda r_0^2}^{\pi\lambda r_{ljk'}^2}\frac{e^{-s}\exp{\bigl(D^\text{asyn}\left(T,n,x\right)\left(\pi\lambda\right)^{-\frac{\alpha\epsilon}{2}}s^{\frac{\alpha\epsilon}{2}}r_{ljk'}^{-\alpha}
\bigr)}}{e^{-\pi\lambda r_0^2}-e^{-\pi\lambda r_{ljk'}^{-\alpha} }}\,\mathrm{d}s
\biggr\}\biggr)^{N_\text{p}}\nonumber\\
%&\overset{(u)}{=}\exp\biggl(
%\int_{R_\text{e}}^\infty 
%\int_{\pi \lambda r_0^2}^{\pi \lambda r^2} \frac{N_\text{p}e^{-s}}{e^{-\pi\lambda r_0^2}-e^{-\pi\lambda r^2 }}
%\left[\exp{\bigl(D^\text{asyn}\left(T,n,x\right)R_\text{e}^{\alpha\epsilon}s^\frac{\alpha\epsilon}{2} r^{-\alpha}
%\bigr)}-1\right]2\pi\lambda r\,\mathrm{d}s
%\,\mathrm{d}r
%\biggr)\nonumber\\
&&\overset{(u)}{=}\exp\biggl(
\int_{\pi\lambda R_\text{e}^2}^\infty 
\int_{\pi \lambda r_0^2}^t \frac{N_\text{p}e^{-s}\left[\exp\bigl(D^\text{asyn}\left(T,n,x\right)R_\text{e}^{-\alpha\left(1-\epsilon\right)}s^\frac{\alpha\epsilon}{2} t^{-\frac{\alpha}{2}}
\bigr)-1\right]}{e^{-\pi\lambda r_0^2}-e^{-t }}
\,\mathrm{d}s\,\mathrm{d}t
\biggr),
\end{eqnarray}
where in $(s)$, we use the distance distribution given in Theorem \ref{thm2}. Then, $(t)$ comes from the variable transform $s=\pi\lambda r_{jjk}^2$. In $(u)$, the Campbell's theory and the exclusion ball model are used. Also, the variable changes as $t=\pi\lambda r^2$.

By following the same procedure as in (\ref{E2_async}) for $\mathcal{E}_2^\text{asyn}$, we have
\begin{eqnarray}
&&\mathcal{E}_2^\text{syn}\left(T,n,x\right)=\mathbb{E}_{\{r_{jjk},r_{ljk}|j\in\mathrm{\Phi}_\text{b}\backslash\{l\}\}\big|r_{llk}=x}\left\{\exp
{\left(D^\text{syn}\left(T,n,x\right)\sum_{j\ne l}
r_{jjk}^{\alpha\epsilon}r_{ljk}^{-\alpha}
\right)}
\right\}\nonumber\\
&&=\exp\biggl(
\int_{\pi\lambda R_\text{e}^2}^\infty 
\int_{\pi \lambda r_0^2}^t \frac{e^{-s}}{e^{-\pi\lambda r_0^2}-e^{-t }}
\left[\exp{\bigl(D^\text{syn}\left(T,n,x\right)R_\text{e}^{-\alpha\left(1-\epsilon\right)}s^\frac{\alpha\epsilon}{2} t^{-\frac{\alpha}{2}}
\bigr)}-1\right]\,\mathrm{d}s
\,\mathrm{d}t\biggr),
\end{eqnarray}
where $D\left(T,n,x\right)$ in both modes are defined in the following
\begin{eqnarray}
\label{d_asyn}
D^\text{asyn}\left(T,n,x\right)&=&-\frac{\eta nT}{C_M^2}\frac{N_\text{p}+N_\text{u}}{N_\text{tot}^2}x^{\alpha\left(1-\epsilon\right)}\left(N_\text{p}+\frac{\sigma^2}{P_\text{d}\omega}x^\alpha\right),\\
\label{d_syn}
D^\text{syn}\left(T,n,x\right)&=&-\frac{\eta nT}{C_M^2}x^{\alpha\left(1-\epsilon\right)}\left(N_\text{p}+\frac{\sigma^2}{P_\text{d}\omega}x^\alpha\right).
\end{eqnarray}

\end{document}